\documentclass{article}

\usepackage{PRIMEarxiv}
\usepackage{amsfonts}       % blackboard math symbols
\usepackage{amscd}          % commutative diagrams
\usepackage[utf8]{inputenc} % allow utf-8 input
\usepackage[T1]{fontenc}    % use 8-bit T1 fonts
\usepackage{hyperref}       % hyperlinks
\usepackage{url}            % simple URL typesetting
\usepackage{booktabs}       % professional-quality tables
\usepackage{amsfonts}       % blackboard math symbols
\usepackage{nicefrac}       % compact symbols for 1/2, etc.
\usepackage{tikz-cd}        % commutative diagrams with tikzcd
\usepackage{bm}
\usepackage{microtype}      % microtypography
\usepackage{lipsum}
\usepackage{fancyhdr}       % header
\usepackage{graphicx}       % graphics
\usepackage{multirow}
\usepackage{amsmath}
\usepackage{subcaption}
\usepackage{array}
\usepackage{graphicx}
\usepackage{threeparttable} 
\usepackage{bm}
\graphicspath{{media/}}     % organize your images and other figures under media/ folder
\usepackage{xcolor} 
\usepackage{color}
\usepackage{enumitem}
\usepackage{algorithm}
\usepackage{algpseudocode}
\usepackage{tabularx}  
\usepackage{mdframed}
\usepackage{tcolorbox}
\usepackage{amsmath,amssymb} % 数学符号依赖
\usepackage{pgfplots}        % 绘图核心包
\pgfplotsset{compat=1.17}    % 兼容pgfplots版本（按需调整）
\usepackage{wrapfig}
\usepackage{tikz}  % 先加载核心绘图包（必须在前）
\usetikzlibrary{arrows.meta}  % 可选：用于美观箭头（如Stealth）
\usetikzlibrary{shapes.geometric}  % 加载几何形状扩展库（支持圆角矩形、平行四边形等）

\newtcolorbox{textbox}{
    colframe=gray!60,    % 灰色边框（60%灰度，可调整数值改变深浅）
    colback=white,       % 无底色（与文档背景一致）
    boxrule=1pt,         % 边框粗细（1pt 较适中）
    left=10pt,           % 左侧内边距
    right=10pt,          % 右侧内边距
    top=6pt,             % 顶部内边距
    bottom=6pt,          % 底部内边距
    sharp corners        % 直角边框（可选，默认是圆角）
}
\newcommand{\qizhi}[1]{\textcolor{black}{#1}}
\newenvironment{Qizhi}
{\color{black}}  % 开始：蓝色
{}              % 结束：无

\newcommand{\wqruan}[1]{\textcolor{black}{#1}}
\newcommand{\linyu}[1]{\textcolor{black}{#1}}

\newcolumntype{P}[1]{>{\centering\arraybackslash}p{#1}}

\usepackage{amssymb, amsthm, amsmath}
\newtheorem{definition}{Definition}
\newtheorem{theorem}{Theorem}

\newtheorem{remark}{Remark}
\newtheorem{corollary}{Corollary}[theorem]

\title{Towards Privacy-Preserving LLM Inference via Covariant Obfuscation (Technical Report)
}

\author{
  Yu Lin, Qizhi Zhang, Wenqiang Ruan, Daode Zhang, Jue Hong, Ye Wu \thanks{The authors thank Prof. Cong Wang of the City University of Hong Kong for his valuable suggestions on revising this paper.}\\
  ByteDance 
  \And
  Hanning Xia, Yunlong Mao, Sheng Zhong \\
  Nanjing University
}

\begin{document}
\maketitle

\begin{abstract}
%-------------------------------------------------------------------------------

The rapid development of large language models (LLMs) has driven the widespread adoption of cloud-based LLM inference services, while also bringing prominent privacy risks associated with the transmission and processing of private data in remote inference. For privacy-preserving LLM inference technologies to be practically applied in industrial scenarios, three core requirements must be satisfied simultaneously: (1) Accuracy and efficiency losses should be minimized to mitigate degradation in service experience. (2) The inference process can be run on large-scale clusters consist of heterogeneous legacy xPUs. (3) Compatibility with existing LLM infrastructures should be ensured to reuse their engineering optimizations. To the best of our knowledge, none of the existing privacy-preserving LLM inference methods satisfy all the above three constraints while delivering meaningful privacy guarantees. In this paper, we propose \texttt{AloePri}, the first privacy-preserving LLM inference method for industrial applications. \texttt{AloePri} protects both the input and output data by covariant obfuscation, which jointly transforms data and model parameters to achieve better accuracy and privacy. We carefully design the transformation for each model component (e.g., Attention, FFN) to ensure inference accuracy and data privacy while keeping full compatibility with existing infrastructures of Language Model as a Service (LMaaS). \texttt{AloePri} has been integrated into an industrial LMaaS system deployed on a multi-node GPU cluster for the evaluation of mainstream LLMs. The evaluation on Deepseek-V3.1-Terminus model (671B parameters) demonstrates that \texttt{AloePri} causes accuracy loss of $0.0\% \sim 3.5\%$ and exhibits efficiency equivalent to that of plaintext inference. Meanwhile, \texttt{AloePri} successfully resists state-of-the-art attacks (e.g., internal state inversion), with less than 5\% of tokens recovered. To the best of our knowledge, \texttt{AloePri} is the first method to exhibit practical applicability to large-scale models in real-world systems.

\end{abstract}

% keywords can be removed

\section{Introduction}
% TODO: 
% 1. 强调实际应用，与AI infra的结合，如kv cache, pd分离。之前的方案不兼容现有的AI infra，例如需要交互等，而本方案实现一体化的隐私保护推理。
% 2. Contribution方案设计思路的两点，并在方案设计部分对应上

In the era of large language models (LLMs), a surge in application demands from clients (e.g. financial companies and hospitals) drives cloud service providers to offer low-cost language model as a service (LMaaS). However, to use LMaaS, clients have to transmit their private data to remote LLM services, bringing severe privacy concerns. While resolving the privacy concerns is important, the following three constraints must be satisfied in an industrial scenario to keep the product competitiveness and a high return on investment (ROI) when designing privacy-enhancement technologies: 
\begin{enumerate}

    \item \textbf{Accuracy and Efficiency Constraint.} Privacy protection methods must not incur perceptible accuracy degradation or efficiency loss. A model performance gap larger than a generational leap (e.g., 88.3\% for Claude 3.5 Sonnet vs. 85.7\% for Claude 3 on MMLU~\cite{anthropic}) correlates with perceptible disparities in client-facing output quality~\cite{hendrycks2021measuring,liang2022holistic,zheng2023judging}. End-to-end latency exceeding 100~ms is perceptually discernible to clients~\cite{miller1968response}. Such accuracy and efficiency impairments would result in an unacceptable decline in product competitiveness.
    
    % \item \wqruan{\textbf{Accuracy and Efficiency Constraint.} Privacy protection methods must not incur significant losses in accuracy and efficiency. A model performance gap equivalent to a generational leap (e.g., GPT-4 vs. GPT-3.5), or an increase in token latency of more than xx ms, will result in an unacceptable decline in product competitiveness.}
    % \item \qw{\textbf{User Experience Constraint. }For commercial systems, privacy-preserving mechanisms must incur only imperceptible overheads.
    % Prior Human--Computer Interaction studies indicate that end-to-end latency beyond approximately 100\,ms becomes perceptually noticeable to users~\cite{miller1968response}.
    % Standardized benchmarks such as MMLU show that contemporary large language models differ by only several percentage points (e.g., 88.7\% of Claude 3.5 Sonnet v.s. 86.8\% of Claude 3~\cite{anthropic}) in absolute accuracy, and these quantitative gaps correlate with perceptible differences in user-facing output quality~\cite{hendrycks2021measuring,liang2022holistic,zheng2023judging}.
    % }

    \item  \textbf{Hardware Compatibility Constraint.} LLM inference clusters of cloud service providers typically consist of heterogeneous legacy xPUs, e.g., Graphics Processing Units (GPUs), Neural Processing Units (NPUs), and Field Programmable Gate Arrays (FPGAs). Therefore, to keep a high ROI, a privacy protection method must be able to run on heterogeneous legacy xPUs clusters.
    
    \item  \textbf{Software Compatibility Constraint.} The method should be compatible with existing LLM infrastructures so that their engineering optimizations can be reused. The current inference frameworks (e.g., vLLM, SGLang) have integrated many software optimizations (e.g., KV Cache, P/D disaggregation). Any privacy protection method that is not compatible with these frameworks would require huge engineering efforts to re-implement these software optimizations, leading to a low ROI.
\end{enumerate}

The above constraints severely restrict the applicability of methods based on cryptography or trusted execution environment (TEE). While cryptographic methods~\cite{gupta2023sigma,zhang2024secure} can provide rigorous privacy guarantees, their substantial computational overhead and incompatibility with LMaaS infrastructures render them impractical for industrial applications. Meanwhile, a large number of existing computing clusters contain legacy xPUs without hardware-native security guarantees. Therefore, TEE-based methods \cite{tan2025pipellm, ng2020goten,10646815} cannot be widely deployed in such scenarios.

While obfuscation-based methods~\cite{tong2025inferdpt, chowdhury2025pr, tang2023privacy, siyan2024papillon, shen2024fire, du2023dp, roberts2025learning, yuan2023secure, nguyen2026noirprivacypreservinggenerationcode} are promising for privacy-preserving LLM inference because they are efficient and broadly compatible with existing infrastructure, these methods are still far from practical usage. 
% Common obfuscation techniques—such as noise injection, transformations, anonymization, and projection—have proven practical and effective in machine learning \cite{al2019privacy, bakken2004data}.

\qizhi{
To quantify privacy guarantees, researchers commonly adopt formal privacy frameworks such as Differential Privacy (DP) \cite{dwork2006differential}, Rényi Differential Privacy (RDP) \cite{mironov2017renyi}, and metric Differential Privacy (mDP) \cite{xie2025decademetricdifferentialprivacy}. Mechanisms for realizing these privacy definitions typically employ either additive noise injection (e.g., the Gaussian and Laplace mechanisms) or random replacement strategies (e.g., the exponential mechanism).
} 
Building on these mechanism, recent work has proposed obfuscation methods for LLM inference that operate at either the token level \cite{tong2025inferdpt, chowdhury2025pr, tang2023privacy, siyan2024papillon, shen2024fire} or the embedding level \cite{du2023dp, roberts2025learning, yuan2023secure, nguyen2026noirprivacypreservinggenerationcode}. Despite their promise, such methods suffer poor trade-off between privacy and accuracy. 
% Generative LLMs depend on rich contextual semantic dependencies between inputs and outputs, which restricts how aggressively data can be obfuscated without degrading accuracy. Moreover, recent inversion attacks show that constrained obfuscation often fails to prevent reconstruction of private data \cite{kugler2021invbert, lin2024inversion, dong2025depth}.
As we show in Section \ref{sec:exp}, several state-of-the-art obfuscation methods \cite{roberts2025learning, du2023dp, tong2025inferdpt, yue2021differential} suffer an accuracy loss of more than 30\% while still being vulnerable to inversion attacks that recover over 50\% of tokens.

In this paper, to achieve better privacy, efficiency, and accuracy, we first propose \textit{Covariant Obfuscation}-a novel privacy-preserving LLM inference mechanism that jointly transforms input data and model weights. By transforming model weights to mitigate the impact of input obfuscations during inference, covariant obfuscation can provide a stronger privacy guarantee than data-only obfuscation methods at equal model accuracy. Furthermore, we introduce three composition theorems (sequential, parallel, and summation) to streamline the construction of covariant obfuscation. 

Second, following covariant obfuscation, we propose \texttt{AloePri}, an \texttt{A}ccurate, \texttt{lo}w-cost, and \texttt{e}fficient solution for \texttt{Pri}vacy-preserving LLM inference. As illustrated in Figure~\ref{fig:overview}, model obfuscation is performed offline: the client generates a secret token-level permutation and applies it to token-correlated model weights, including the embedding layer and model head. Additional Gaussian noise and carefully selected invertible matrices are further employed to transform the weights of all model components (e.g., attention, FFN) to hide the secret permutation while preserving model inference accuracy. During online inference, the client obfuscates its prompts before sending them to the server and decodes the returned responses using the secret permutation. Under this covariant obfuscation, the server observes only obfuscated inputs and outputs and thus cannot access the client’s private data.

\begin{figure}[t]
\centering 
\includegraphics[width=0.6\textwidth]{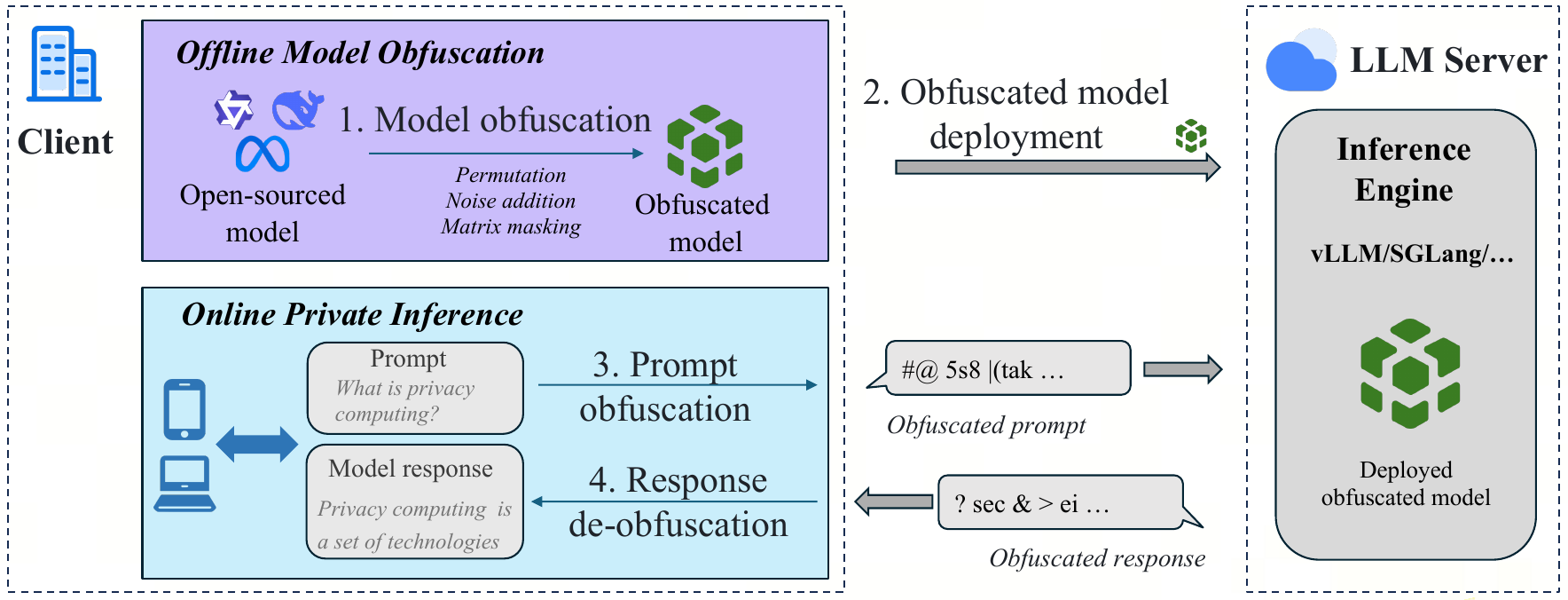}
\caption{Workflow of \texttt{AloePri}. The client locally obfuscates the model and deploys it to the server, and handles prompt obfuscation and response de-obfuscation during online phase.}\label{fig:overview}
\end{figure}

% \wqruan{To quantify the privacy guarantees provided by \texttt{AloePri}, we propose a new privacy notion by extending classical Probably Approximately Correct (PAC) Privacy~\cite{DBLP:conf/crypto/XiaoD23}, namely extended PAC privacy. Based on extended PAC privacy, we quantitatively characterize how model properties (e.g., number of weights, parameter dimension, and vocabulary size) and privacy weights affect \texttt{AloePri}'s privacy guarantee. The quantitative relationship not only shows that \texttt{AloePri} can effectively protect users' private information, but also provides a useful guideline for parameter selection of \texttt{AloePri}.  }

We evaluate \texttt{AloePri} via theoretical analysis and empirical experiments. Theoretical analysis characterizes the computational error introduced by \texttt{AloePri}, establishes upper bounds on its information leakage, and quantifies the corresponding attack \wqruan{success rate}. On the other hand, we conduct comprehensive empirical evaluations on mainstream open-source LLMs (including Qwen2.5/Qwen3, Llama3, Qwen3-MoE, Deepseek-R1-Distill-Qwen, and Deepseek-V3.1-Terminus) and test \texttt{AloePri}'s resilience against representative attacks~\cite{dong2025depth, thomashidden, kugler2021invbert, aldarrab-may-2021-sequence, kambhatla2023decipherment} on privacy-preserving LLM inference methods. Experimental results show that \texttt{AloePri} significantly mitigates privacy leakage while retaining both model accuracy and inference efficiency. For instance, on Deepseek-V3.1-Terminus, \texttt{AloePri} incurs only  $0\% \sim 3.5\%$ accuracy loss, with fewer than 5\% of tokens recovered by several state-of-the-art attacks~\cite{dong2025depth, thomashidden, kugler2021invbert, aldarrab-may-2021-sequence, kambhatla2023decipherment}. Furthermore, as its inference pipeline remains nearly identical to plaintext inference, \texttt{AloePri} supports direct integration into mainstream inference frameworks (e.g., vLLM \cite{kwon2023efficient} and SGLang \cite{zheng2024sglang}) with negligible efficiency overhead and engineering efforts.

We highlight our contributions as follows:
\begin{itemize}
    \item We propose covariant obfuscation, a novel privacy-preserving LLM inference mechanism. We formalize its compositional properties, enabling the seamless integration into sophisticated LLM architectures. (Section~\ref{sec:cov_obfus}) 
    \item Based on the covariant obfuscation mechanism, we propose \texttt{AloePri} that achieves efficiency comparable to plaintext model inference, provides robust privacy protection on both input and output data, and incurs slight accuracy loss. Notably, \texttt{AloePri} does not change the model structure and is compatible with existing LLM infrastructures. (Section~\ref{AloePri})
    \item \linyu{We introduce a novel Rényi-metric differential privacy definition tailored for analyzing the privacy guarantees of high-dimensional Gaussian noise. Building on this framework, we theoretically prove that \texttt{AloePri} achieves a stronger privacy guarantee than data-only obfuscation by incorporating obfuscation on model weights. (Section~\ref{sec:dp_security})}
    \item We conduct comprehensive experiments to validate the effectiveness of \texttt{AloePri}. Experimental results across multiple datasets and LLMs demonstrate that \texttt{AloePri} significantly outperforms state-of-the-art obfuscation methods in both data privacy and model inference accuracy. (Section~\ref{sec:exp})
\end{itemize}

\section{Related Works}
\label{sec:related}
\subsection{Methods based on Obfuscation}
A series of obfuscation-based methods have been proposed for privacy-preserving LLM inference, primarily via three techniques: token substitution \cite{tong2025inferdpt, yue2021differential, chowdhury2025pr, shen2024fire}, embedding transformation \cite{du2023dp, yuan2023secure, roberts2025learning, zeng-etal-2025-privacyrestore, mai2023split, nguyen2026noirprivacypreservinggenerationcode}, LLM-aided rewriting \cite{siyan2024papillon, lin2025emojiprompt}.

\textbf{Token Substitution.} Since prompts are composed of a sequence of tokens, the most intuitive privacy preservation method is to substitute sensitive tokens within the text. SANTEXT \cite{yue2021differential} and RANTEXT \cite{tong2025inferdpt} leverage local differential privacy to replace sensitive tokens with alternatives based on token embedding similarity. ProSan \cite{shen2024fire} quantifies word importance and privacy risk to dynamically balance accuracy and privacy. However, these methods only protect a limited set of sensitive tokens and lack semantic-level protection. Moreover, if substituted tokens correlate with task objectives, they undermine the consistency between input text and task intent.

\textbf{Embedding Transformation.} Some studies obfuscate the token embeddings of input data to achieve semantic-level privacy protection. SGT \cite{roberts2025learning} learns an embedding transformation model via a mutual information-based training loss, converting raw user inputs into noisy embeddings for secure remote LLM inference. DP-Forward \cite{du2023dp} designs a differential privacy mechanism to set privacy budgets for high-dimensional embedding perturbation, balancing privacy and inference accuracy. STIP \cite{yuan2023secure} adopts feature-space random permutations to enable efficient private inference in a three-party scenario. These methods require the client and server to transmit data in the form of embeddings, which increases communication overhead. Meanwhile, existing inversion attacks \cite{kugler2021invbert, lin2024inversion, dong2025depth, thomashidden} pose a critical threat, as adversaries can invert obfuscated embeddings to recover the original private texts. 
In addition, following the framework of split inference, some studies deploy part of models at the client side to process and obfuscate data. SnD \cite{mai2023split} splits LLMs into local/cloud encoders, such that the client can add noise to hidden states outputted by local encoders and use pre-trained denoiser on cloud encoders' outputs to obtain high-quality inference results. NOIR \cite{nguyen2026noirprivacypreservinggenerationcode} splits open-source LLMs into a lightweight client-side encoder/decoder and a cloud-hosted middle layer, which are trained jointly to achieve accuracy and privacy. Meanwhile, it introduces the concept of indistinguishability-preserving vocabulary, which requires the client to locally randomize token embeddings and process them with the client-side encoder. While split inference enables the client to better balance accuracy and privacy, it requires the client to participate in a training process and take part of forward computational model layers during inference, which is not compatible with the existing LLM infrastructure.

\textbf{LLM-aided Rewriting.} This type of method leverages LLMs themselves to hide private information in input data. PAPILLON \cite{siyan2024papillon} deploys a multi-stage LLM pipeline (local + internet models) to preserve privacy while maintaining high task response quality. EmojiPrompt \cite{lin2025emojiprompt} leverages cloud-hosted LLMs to rewrite private texts with special symbols and emojis. These methods rely on LLMs deployed in a trusted local or third-party environment to perform private text rewriting, and quantifying the privacy guarantees of the rewritten texts remains a significant challenge.

% \textbf{Embedding-level Obfuscation Methods.}  Different from token-level obfuscation methods, embedding-level obfuscation methods directly add noise to the embeddings, ensuring that users’ private text data is not directly exposed to the server. For example, DP-Forward \cite{du2023dp} proposes a novel differential privacy mechanism for setting privacy budgets that balance privacy and accuracy in high-dimensional embedding perturbation. STIP \cite{yuan2023secure} is built on a three-party scenario and adopts feature-space random permutations to realize efficient and private inference. SGT \cite{roberts2025learning} includes a mutual information-based training loss function to learn an embedding transformation model. The transformation model converts raw user inputs into noisy embeddings that can be safely used for LLM inference on remote servers. Nevertheless, existing inversion attacks \cite{kugler2021invbert, lin2024inversion, dong2025depth, thomashidden} remain a significant threat to these embedding-level methods. That is, attackers can invert the obfuscated embeddings to recover the original user inputs. In our experiments, we extend these attacks according to the characteristics of \texttt{AloePri} to evaluate the privacy level of \texttt{AloePri} comprehensively.

\subsection{Methods based on Cryptography and TEE}
Cryptographic primitives provide rigorous privacy guarantees, spurring many cryptography-based methods \cite{gupta2023sigma, zhang2024secure,BumbleBee}. Gupta et al.~\cite{gupta2023sigma} combine additive secret sharing (ASS) and function secret sharing for efficient online inference at the cost of massive offline computation and communication overhead. Zhang et al.~\cite{zhang2024secure} implement one-round inference via fully homomorphic encryption (FHE). Some studies \cite{BumbleBee} combine ASS and FHE to balance communication and computation overhead. Li et al.~\cite{mpcformer} optimize such methods by modifying activation functions. Despite recent advances, these methods remain far from practical. For example, BumbleBee~\cite{BumbleBee} takes about 8 minutes to generate a single token for LLaMA-7B.

TEE-based methods constitute another key route for privacy-preserving LLM inference. PipeLLM \cite{tan2025pipellm} and ccAI \cite{10.1145/3725843.3756104} leverage GPU-TEE to enable confidential LMaaS. By adopting speculative pipelined encryption, PipeLLM limits the additional throughput overhead to less than 20\% for LLMs with parameter sizes spanning 13B to 175B.  ccAI adds an additional FPGA-based PCI link to secure PCIe packet transmission across diverse xPU types without requiring modifications to applications or drivers. While these methods have promising efficiency, they are either not suitable for heterogeneous legacy xPU clusters or require a huge hardware retrofit cost. Additionally, several approaches \cite{tramer2018slalom,ng2020goten,10646815} have been proposed to leverage CPU-TEE and legacy GPUs for privacy-preserving inference. However, frequent data transfers between CPU-TEE and legacy GPUs render these methods incompatible with key optimizations in modern LLM infrastructure (e.g., KV Cache, P/D disaggregation).

\section{Preliminaries}

\subsection{System Model}
We focus on the general scenario of LMaaS. A \textit{client} holding private data aims to access the LLM inference service provided by a cloud \textit{server}. The client is only capable of undertaking a certain level of computational overhead during the offline phase. Meanwhile, during the online inference phase, to avoid compromising inference throughput, they can only handle lightweight computational tasks. The LLM inference service provided shall be compatible with the interfaces of modern LLM frameworks, i.e., the service takes prompts as inputs and outputs inference results via one-round communication. For instance, a software company requests a remote LLM server to use code generation agents.

\subsection{Threat Model}
Our goal is to protect the private prompt and model response during LLM inference. We assume the server is honest-but-curious: it faithfully provides the inference service but actively seeks to learn the client's private information. The attacker can access the model and observe the service's inputs and outputs. 
% In addition, the attacker possesses the following capabilities
\wqruan{For example, the attacker can try recovering users' private data through the following three methods}: 1) \textbf{Obfuscation recovery}: the attacker knows the obfuscation mechanism and tries to breach it by exploiting its structural weaknesses \cite{lin2024inversion, thomashidden}. 2) \textbf{Training-based inversion}: the attacker can use substantial computational resources to recover private data via dedicated training-based inversion \cite{dong2025depth, kugler2021invbert}. 3) \textbf{Token-frequency exploitation}: the attacker continuously monitors input and output tokens and records their frequencies over time to infer private information \cite{aldarrab-may-2021-sequence}.

\subsection{Notations}
We use $\mathcal{V}$ and $\theta$ to represent the vocabulary and model weights, respectively. $n$ denotes the number of tokens, and $d$ represents the hidden size. Throughout this paper, $S_{n}$ is defined as the permutation group over $[1, n]$, $O_d$ is the orthogonal group of $d\times d$ matrices, and $\text{id}_{A}$ denotes the identity transformation on the space $A$. Noised data is denoted by $\square^{\star}$, while randomly sampled data is marked with $\hat{\square}$, and $\widetilde{\square}$ denotes obfuscated model weights. 
% We use $\mathcal{I}(x; y)$ to denote the mutual information between two random variables $x$ and $y$, and $H(x)$ represents the information entropy of $x$. 
We summarize all notations used in this paper in Appendix \ref{appx:notation}.

\subsection{Large Language Model}
\label{sec:llm_structure}
\subsubsection{Model Structure}
A typical LLM consists of a vocabulary $\mathcal{V}$ with $n$ tokens, as well as five key components: an embedding layer, a model head, $L$ attention layers, $L$ feed-forward network (FFN) layers, and multiple normalization layers.

\textbf{Vocabulary, Embeddings, and Model Head.} The vocabulary $\mathcal{V}$ and corresponding merge rules serve to tokenize textual prompts into a sequence of token indices. The weight matrix of the embedding layer $W_{e}$ is an $n\times d$-dimensional matrix, and that of the model’s head layer $W_{h}$ is a $d\times n$-dimensional matrix. Each token index in $\mathcal{V}$ corresponds to a $d$-dimensional vector in $W_{e}$ and $W_{h}$, respectively.

\textbf{FFN}. The structures of FFN vary between two types of LLMs: dense models and Mixture-of-Experts (MoE) models \cite{lepikhin2020gshard}. In dense models, each FFN layer consists of three sets of weights, denoted as $\omega_{\text{ffn}} = (W_{\text{gate}}, W_{\text{up}}, W_{\text{down}})$. Given the hidden states $x$ as input, the forward pass output of the FFN can be calculated as:
\begin{equation}
\label{ffn}
    f_{\text{ffn}}(x, \omega_{\text{ffn}}) = \left(
\text{SiLU}(x W_{\text{gate}}) \odot (x W_{\text{up}})
\right) W_{\text{down}},
\end{equation}
where $\text{SiLU}$ is a commonly used activation function in LLMs, and $\odot$ denotes the Hadamard product. In MoE models, each expert corresponds to an FFN layer with the same structure as that in dense models. Furthermore, an expert router (denoted as $W_{\text{router}}$) is employed to select activated experts during model inference.

\textbf{Attention Layer}.
In this paper, we mainly consider four popular attention mechanisms: Multi-head Attention (MHA) \cite{vaswani2017attention}, Multi-Query Attention (MQA) \cite{shazeer2019fast}, Grouped-Query Attention (GQA) \cite{ainslie2023gqa}, and Multi-head Latent Attention (MLA) \cite{deepseekai2025deepseekr1incentivizingreasoningcapability}. The weights of the attention layer can be denoted as $\omega_{\text{attn}} = (W_{\text{q}}, W_{\text{k}}, W_{\text{v}}, W_{\text{o}})$. We define the total number of attention heads as $m$, and the number of heads corresponding to the attention key and value weights as $m_{\text{kv}}$. For example, $\frac{m}{m_{\text{kv}}} > 1$ for GQA, $\frac{m}{m_{\text{kv}}} = 1$ for MHA, and $\frac{m}{m_{\text{kv}}} = m$ for MQA. Taking GQA as an example, the attention weights can be split by the number of heads as: \(
    W_{\text{q}} = \left[ \begin{matrix} W^{(1)}_{\text{q}}  & \dots &W^{(m)}_{\text{q}} \end{matrix} \right], 
    W_{\text{k}} = \left[ \begin{matrix} W^{(1)}_{\text{k}} & \dots & W^{(m_{kv})}_{\text{k}} \end{matrix} \right], 
    W_{\text{v}} = \left[ \begin{matrix} W^{(1)}_{\text{v}} & \dots & W^{(m_{kv})}_{\text{v}} \end{matrix} \right], 
    W_{\text{o}} = \left[ \begin{matrix} W^{(1)}_{\text{o}} & \dots &W^{(m)}_{\text{o}} \end{matrix} \right]^T.
\)
For the input hidden states $x$, the attention output of GQA is computed as
\begin{align}
\label{eq:gqa}
    & f_{\text{attn}}(x, \omega_{\text{attn}})
     =  \sum_{i}
    g\left(
    \frac{
    \mathcal{G}(xW^{(i)}_{\text{q}}) \mathcal{G}(x W^{\eta(i)}_{\text{k}})^T
    }{
    \sqrt{d_{\text{head}}}
    }
    \right) x W^{\eta(i)}_{\text{v}}
     W^{(i)}_{out}.
\end{align}

In the above equation, $i$ denotes the head index for $W_{\text{q}}$ and $W_{\text{o}}$, while $\eta(i)$ represents the corresponding head index for $W_{\text{k}}$ and $W_{\text{v}}$. $g$ denotes the softmax function. $\mathcal{G}$ stands for the positional embedding function, such as Rotary Position Embeddings (RoPE)~\cite{su2024roformer}.
% The multiplication of attention query and key is also called \textit{attention score}. 

% As for MLA, it uses low-rank weights for attention query and key to reduce computational and memory overhead. Meanwhile, it also targets the characteristics of rotary positional encoding, optimizing its integration with the low-rank attention mechanism to maintain effective positional information modeling while ensuring efficiency.

\textbf{Layer Normalization}. 
In this paper, we primarily focus on the Root Mean Square Layer Normalization (RMSNorm), the most commonly used layer normalization method in modern LLMs. RMSNorm utilizes a weight parameter $w_{\text{norm}}$ and the root mean square to normalize an input state $x$ with dimension $d$, which is calculated as:
$\text{RMSNorm}(x) = \frac{x \odot w_{\text{norm}}}{\sqrt{\frac{1}{d}\sum_{i=1}^{d} x_i^2}}$.

\subsubsection{Workflow of Auto-regressive Text Generation}

During inference, the model accepts a textual prompt as input. Using the vocabulary $\mathcal{V}$, the prompt is first tokenized into a sequence of $l$ tokens $\mathcal{T} =  \{tok_1, \dots, tok_l\}$, which are further converted to a token index sequence $x = \{x_i\}_{i=1}^l$. Let $\theta$ denote the set of all model weights. The model’s forward computation can be formulated as $f(x, \theta) = y$, which yields a token index $y \in \mathbb{Z}_n$. Specifically, the $l$ token indices first retrieve the corresponding $l$ embeddings via $W_{e}$. These embeddings are then fed into the attention and FFN layers to extract semantic information. Finally, the extracted features are used to compute logits with $W_{h}$, from which the index of the next generated token $y$ is sampled. To enable auto-regressive generation, the generated token index $y$ is appended to the original token index sequence $x$ to serve as the input for the subsequent inference step.

% \subsection{Mutual Information and PAC Privacy}
% \begin{definition}[$(\delta, \rho, D)$-PAC Privacy]
% \label{def:pac}
% For a data distribution $D$ over a space $\mathcal{X}$, an obfuscation mechanism $\mathcal{M}: \mathcal{X} \rightarrow \mathcal{X}^{\star}$, and a measure function $\rho$, $\mathcal{M}$ is said to satisfy $(\delta, \rho, D)$-PAC Privacy if it holds that: There does not exist an adversary that can recover an approximation $\tilde{X}$ of X with a probability higher than $1-\delta$ so that $\rho(\tilde{X}, X) = 1$ given only $D$, $\mathcal{M}$, and $\mathcal{M}(X)$.
% \end{definition}

% PAC Privacy \cite{xiao2023pac} models the hardness of data reconstruction for an obfuscation mechanism $\mathcal{M}$. The adversary's probability of failure, denoted by $\delta$, and subject to a specific measure function $\rho$, serves as an effective metric for the security level of the given scenario. The measure function $\rho$ can be tailored to different threat scenarios, depending on the adversary's target and the server's security requirements.

% Let $\delta_0$ be the prior failure probability without observing any private information. The Kullback-Leiber Divergence between $\delta_0$ and $\delta$ is bounded by the mutual information between $X$ and $\widetilde{X}$:
% \begin{equation}
%     \Delta_{KL} \delta \leq \mathcal{I}(X; \widetilde{X}),
% \end{equation}
% where $\Delta_{KL}\delta = \delta \log(\frac{\delta}{\delta_0}) + (1 - \delta) \log(\frac{1 - \delta}{1 - \delta_0})$.

\section{Covariant Obfuscation}
\label{sec:cov_obfus}
% This section elaborates on covariant obfuscation, a more practical obfuscation mechanism aimed at optimizing the utility-privacy trade-off during model inference.

\linyu{In this section, we first present the definition of covariant obfuscation.
% and demonstrate its advantages in mitigating information leakage compared to traditional data-only obfuscation in LLM inference scenarios. 
We then introduce three composition theorems that allow us to separately design covariant obfuscation for each component of the LLM, and subsequently combine these individual covariant obfuscations as a covariant obfuscation for the whole model.}

\subsection{Definition and Property}

\subsubsection{Covariant Obfuscation on Inference Function}
We consider an inference function \( f: X \times \Theta \to Y \), where \( X\) is the data space, \( \Theta \) is the parameter space, and \( Y \) is the prediction space, and $\tilde{X}, \tilde{Y}, \tilde{\Theta}$ are their obfuscated spaces respectively. A covariant obfuscation $C$ with obfuscation error (a constant) $e_C$ for the function $f$ can be represented as the following quintuple \( (\phi_X, \phi_\Theta, \phi_Y, \psi_Y, \tilde{f}) \): (1) Data obfuscation: \( \phi_X: X \to \tilde{X} \). (2) Model transformation: \( \phi_\Theta: \Theta \to \tilde{\Theta} \). (3) Label obfuscation: \( \phi_Y: Y \to \tilde{Y} \). (4) Label de-obfuscation \( \psi_Y: \tilde{Y} \to Y \). (5) Inference in the obfuscated space: \( \tilde{f}: \tilde{X} \times \tilde{\Theta} \to \tilde{Y} \).
They satisfy the following conditions:

\begin{enumerate}
    \item[ A. ] Commutation condition: 
    % \( \widetilde{f}(\phi_X(x), \phi_\Theta(\theta)) \approx \phi_Y \circ f(x, \theta) \)
    \wqruan{The operations of the obfuscation and the inference function shall be commutative. That is, } $\mathbb{E}[d(\tilde{y},  \phi_Y \circ f(x, \theta))]\leq e_C $, where $\tilde{y}=\tilde{f}(\phi_X(x), \phi_\Theta(\theta))$, and $d(\cdot, \cdot)$ is a distance function on $\tilde{Y}$ (e.g. Euclid Distance), indicating that the error incurred by the obfuscation is bounded by $e_C$.
    % for any \( x \in X \) and \( \theta \in \Theta \).
    \[
    \scriptsize
    \begin{CD}
    X \times \Theta @>f>> Y \\
    @V(\phi_X, \phi_\Theta) VV @V\phi_Y VV \\
    \tilde{X} \times \tilde{\Theta} @>\tilde{f}>> \tilde{Y}
    \end{CD}
    \]
    \item[ B. ] De-obfuscation condition: \( \psi_Y \circ \phi_Y = \text{id}_Y \), where $\text{id}_Y$ denotes the identity transformation on $Y$.
\end{enumerate}

Notably, when \(\tilde{X} = X\), \(\tilde{\Theta} = \Theta\), \(\tilde{Y} = Y\), \(\phi_\Theta = \text{id}_\Theta\), and \(\phi_Y = \psi_Y = \text{id}_Y\), the covariant obfuscation degenerates to a data-only obfuscation given by \(\phi_X: X \to X\).

\subsection{Composition Theorems}
\label{sec:composition_theorem}

% To design covariant obfuscation for general model inference functions, we aim to first design covariant obfuscation for some basic modules, and then construct the covariant obfuscation of the entire model inference function through the composition of these modules. This requires covariant obfuscation to satisfy the module composition property: the composition of covariant obfuscations remains a covariant obfuscation.

% It is noted that the process of model inference can be decomposed into the series composition, parallel composition, and summation of modules (e.g., residual links). Below, we discuss the composition property of covariant obfuscation for these three cases.

% We present the composition theorems of covariant obfuscation and show their proofs in Appendix \ref{appx:proof_cov_obfus_the}.
As LLM includes complicated structures, we formalize the compositional theorems (sequential, parallel, and summation) of covariant obfuscation to enable the design of covariant obfuscation for basic model components. The proofs of the theorems are presented in Appendix \ref{appx:proof_of_comp_the}.

% for covariant obfuscation on two connected functions $f$ and $g$, e.g. the sequential pattern for $f\circ g$,  the parallel pattern for $f|| g$, and the summation pattern for $f+g$.

%\linyu{We show that covariant obfuscation satisfies sequential composition (Theorem \ref{the:sequantial}), parallel composition (Theorem \ref{the:parallel}), and summation composition (Theorem \ref{the:summation}). These theorems enable us to decompose complex LLM inference into small components (e.g., attention $\omega_{\text{attn}}$, FFN $\omega_{\text{ffn}}$), and construct covariant obfuscation for each component separately.}

% \begin{theorem}
% If $(\phi_X, \phi_\Theta, \psi _Y, \tilde{f})$ is a covariant obfuscation of $f: X \times \Theta  \to Y$, $(\phi _Y, \phi _\Xi, \psi _Z, \tilde{g})$ is a  covariant of $g: Y \times \Xi  \to Z$, we can construct a covariant obfuscation of $X \times (\Theta \times \Xi) \to Z$ as follows:
% \[
% \begin{array}{rlcrcl}
% f: & X & \times& (\Theta \times \Xi )& \longrightarrow  & Y \\
% & \downarrow ^ {\phi_X}  & & \downarrow  ^ {(\phi_\Theta, \phi_\Xi)} && \uparrow  ^ {\psi _Z} \\
% \tilde{f}: & \tilde{X} & \times& (\tilde{\Theta} \times \tilde{\Xi}) & \longrightarrow & \tilde{Y}
%  \end{array}
% \]
% \end{theorem}

\subsubsection{Sequential Composition}
We show that two sequentially connected covariant obfuscations designed for two LLM adjacent components (e.g., $W_{\text{v}}$ and $W_{\text{o}}$) still constitute a covariant obfuscation, whose obfuscation error is controlled by the obfuscation errors of two original covariant obfuscations.

\begin{theorem}[Sequential Composition Theorem $f \circ g$]
\label{the:sequantial}
Assume that a covariant obfuscation \( C_1 \) for \( f: X \times \Theta \to Y \), and a covariant obfuscation \( C_2 \) for \( g: Y \times \Xi \to Z \) satisfying the boundary condition: \( {\phi_Y}_{|C_1} = {\phi_Y}_{|C_2} \) (denoted as \( \phi_Y \) uniformly below). Then, \( C_2 \circ C_1 := (\phi_X, (\phi_\Theta, \phi_\Xi), \phi_Z, \psi_Z, \tilde{h}) \) is a covariant obfuscation for \( h: X \times (\Theta \times \Xi) \to Z \) where \( h(x, (\theta, \xi)) := g(f(x, \theta), \xi) \). 
%such that \( \tilde{h}(\tilde{x}, (\tilde{\theta}, \tilde{\xi})) := \tilde{g}(\tilde{f}(\tilde{x}, \tilde{\theta}), \tilde{\xi}) \).

% \[
% \scriptsize
% \xymatrix{
% X \times \Theta \ar[r]^f \ar[d]_{(\phi_X, \phi_\Theta)} & Y \ar[d]^{\phi_Y} \\
% \tilde{X} \times \tilde{\Theta} \ar[r]_{\tilde{f}} & \tilde{Y}
% \ar@{}[ul]|{C_1} % 直接在中间位置放文字，无语法风险
% }
% \quad 
% \xymatrix{
% Y \times \Xi \ar[r]^g \ar[d]_{(\phi_Y, \phi_\Xi)} & Z \ar[d]^{\phi_Z} \\
% \tilde{Y} \times \tilde{\Xi} \ar[r]_{\tilde{g}} & \tilde{Z}
% \ar@{}[ul]|{C_2} % 中间嵌入文字，不破坏箭头结构
% }
% \quad 
% \xymatrix{
% X \times (\Theta \times \Xi) \ar[r]^{h} _ {\underset{C_2 \circ C_1}{}} \ar[d]^{(\phi_X, (\phi_\Theta, \phi_\Xi))} & Z \ar[d]^{\phi_Z} \\
% \tilde{X} \times (\tilde{\Theta} \times \tilde{\Xi}) \ar[r]_{\tilde{h}} & \tilde{Z}
% % \ar@{}[ul]|{C_2 \circ C_1} % 中间居中嵌入文字，不影响箭头和节点
% }
% \]

Furthermore, let \( e_{C_1} \) and \( e_{C_2} \) be the obfuscation errors of \( C_1 \) and \( C_2 \), respectively. For all \( \tilde{y}_1, \tilde{y}_2 \in \tilde{Y} \) and \( \tilde{\xi} \in \tilde{\Xi} \), if \( \tilde{g}(\tilde{y}, \tilde{\xi}) \) satisfies the Lipschitz condition: 
% with respect to \( \tilde{y} \):
\[
d(\tilde{g}(\tilde{y}_1, \tilde{\xi}), \tilde{g}(\tilde{y}_2, \tilde{\xi})) \leq M_g \cdot d(\tilde{y}_1, \tilde{y}_2),
\]
where $M_g$ is the Lipschitz constant for $\tilde{g}$. The obfuscation error of the series composition satisfies:
\[
e_{C_2 \circ C_1} \leq M_g \cdot e_{C_1} + e_{C_2}.
\] 
\end{theorem}

\subsubsection{Parallel Composition}
Some LLM components are connected in parallel, e.g., attention queries and keys. When the covariant obfuscations of these components are connected in parallel, they also constitute a covariant obfuscation.
% \qizhi{which obfuscation error is controlled by the obfuscation errors of two origin  covariant obfuscations.}

\begin{theorem}[Parallel Composition Theorem $f || g$]
\label{the:parallel}
Assume that a covariant obfuscation \( C_1 \) for \( f: X \times \Theta \to Y \) and a covariant obfuscation \( C_2 \) for \( g: X \times \Xi \to Z \) satisfying the boundary condition: \( \phi_{X_{|C_1}} = \phi_{X_{|C_2}} \) (denoted as \( \phi_X \) uniformly below). Then, \( C_1 || C_2 := (\phi_X, (\phi_\Theta, \phi_\Xi), (\phi_Y, \phi_Z), (\psi_Y, \psi_Z), \tilde{h}) \) is a covariant obfuscation for \( h: X \times (\Theta \times \Xi) \to Y \times Z \), where \( h(x, (\theta, \xi)) := (f(x, \theta), g(x, \xi)) \).
%such that \( \tilde{h}(x, (\theta, \xi)) := (\tilde{f}(x, \theta), \tilde{g}(x, \xi)) \).

% \[
% \scriptsize
% \begin{CD}
% X \times \Theta @>f>> Y \\
% @V\phi_X VV       @V\phi_Y VV \\
% \tilde{X} \times \tilde{\Theta} @>\tilde{f}>> \tilde{Y}
% \end{CD}
% \quad \quad \quad \quad \quad \quad 
% \begin{CD}
% X \times \Xi @>g>> Z \\
% @V\phi_X VV       @V\phi_Z VV \\
% \tilde{X} \times \tilde{\Xi} @>\tilde{g}>> \tilde{Z}
% \end{CD}
% \]

Furthermore, let \( e_{C_1} \) and \( e_{C_2} \)be the obfuscation errors of \( C_1 \) and \( C_2 \), if the distance function \( d_{Y \times Z} \) on \( Y \times Z \) satisfies the control condition with respect to the distance functions \( d_Y \) on \( Y \) and \( d_Z \) on \( Z \):
\[
d_{Y \times Z}((y_1, z_1), (y_2, z_2)) \leq d_Y(y_1, y_2) + d_Z(z_1, z_2),
\]
then \( e_{C_1 || C_2} \leq e_{C_1} + e_{C_2} \). 
\end{theorem}

\subsubsection{Summation Composition} 
The covariant obfuscations for bypass computations, e.g., residual connections, can also be composed.
\begin{theorem}[Summation Composition Theorem $f + g$]
\label{the:summation}
Assume that a covariant obfuscation \( C_1 \) for \( f: X \times \Theta \to Y \) and a covariant obfuscation \( C_2 \) for \( g: X \times \Xi \to Y \) satisfying the boundary condition: $\phi_{X_{|C_1}} = \phi_{X_{|C_2}}, \phi_{Y_{|C_1}} = \phi_{Y_{|C_2}}, \psi_{Y_{|C_1}} =  \psi_{X_{|C_2}}$ (denoted as $\phi_X, \phi_Y, \psi_Y$ uniformly below), where \( Y \) and \( \tilde{Y} \) are Abelian groups, and \( \phi_Y: Y \to \widetilde{Y} \) is a group homomorphism. Then,  \( C_1 + C_2 := (\phi_X, (\phi_\Theta, \phi_\Xi), \phi_Y, \psi_Y, \tilde{h}) \) is a covariant obfuscation for for \( h: X \times (\Theta \times \Xi) \to Y \) where \( h(x, (\theta, \xi)) := f(x, \theta) + g(x, \xi) \) such that \( \tilde{h}(x, (\theta, \xi)) := \tilde{f}(x, \theta) + \tilde{g}(x, \xi) \).

% \[
% \scriptsize
% \begin{CD}
% f: X \times \Theta @>>> Y \\
% @V(\phi_X, \phi_\Theta) VV @V\phi_Y VV \\
% \tilde{f}: \tilde{X} \times \tilde{\Theta} @>>> \tilde{Y}
% \end{CD}
% \quad \quad \quad \quad 
% \begin{CD}
% g: X \times \Xi @>>> Y \\
% @V(\phi_X, \phi_\Xi) VV @V\phi_Y VV \\
% \tilde{g}: \tilde{X} \times \tilde{\Xi} @>>> \tilde{Y}
% \end{CD}
% \]

Furthermore, let \( e_{C_1} \) and \( e_{C_2} \) be the obfuscation errors of \( C_1 \) and \( C_2 \), if the distance function on \( Y \) satisfies translation invariance, i.e.,
\[
d(x, y) = d(x + z, y + z),
\]
then \( e_{C_1 + C_2} \leq e_{C_1} + e_{C_2} \).   

\end{theorem}

\section{AloePri}
\label{AloePri}
% 能够明确说明：加噪为什么能保护token、加噪如何确保对效果和安全性的影响
% 故事讲述逻辑：token保护的出发点：置换，置换为什么不影响推理 -> 置换嵌入在参数中，如何保护置换？ 混淆和加噪  -> 混淆和加噪如何不影响模型可用性? 混淆可逆，加噪对效果影响小

In this section, we first present an overview of \texttt{AloePri}, then describe the covariant obfuscation design, as well as the analysis on accuracy and information leakage in detail.

\begin{figure*}[htbp]
    \centering
    \includegraphics[width=0.9\textwidth]{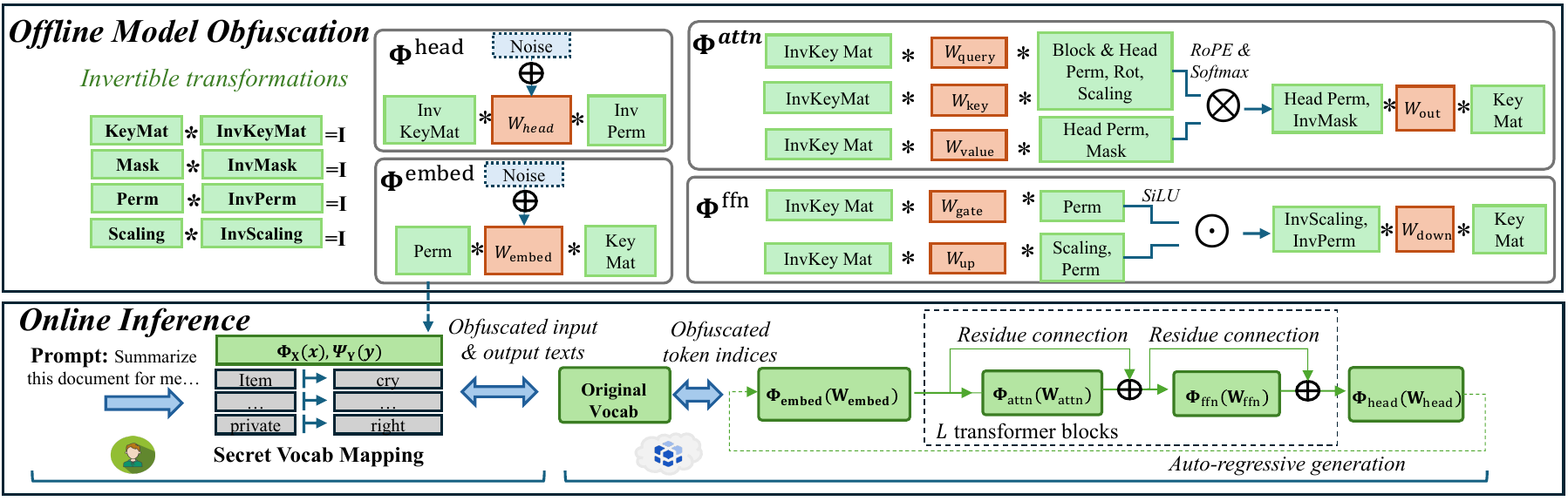}
    \caption{Overview of \texttt{AloePri}. In the offline model obfuscation process, token-level permutation and linear transformations are employed to construct the obfuscations $\phi^{\text{embed}}, \phi^{\text{head}}, \phi^{\text{attn}},$ and $\phi^{\text{ffn}}$. In the online inference process, a secret vocabulary mapping associated with the permutation is used to construct the data obfuscation $\phi_X$ and de-obfuscation $\psi_Y$.}
    \label{fig:AloePri}
\end{figure*}

\subsection{Overview}
As shown in Figure~\ref{fig:AloePri}, based on the covariant obfuscation mechanism, \texttt{AloePri} includes both data obfuscation $\phi_X$ and model obfuscation $\phi_\Theta$ mechanisms for an LLM inference function $f: \mathbb{Z}^{l}_{n} \times \Theta \rightarrow  \mathbb{Z}_{n}$. In an offline model obfuscation process, the client generates a secret permutation $\tau$, which is used in $\phi_X=\tau(x), \phi_Y(y)=\tau(y)$ to permute input and output token indices $x, y$. The permutation $\tau$ is also used to shuffle the embedding and model head at the token level. To prevent the attacker from recovering $\tau$ by comparing plaintext and obfuscated weights of the embedding and model head, a series of sub-obfuscations, including $\phi^{\text{embed}}$, $\phi^{\text{head}}$, $\phi^{\text{attn}}$, and $\phi^{\text{ffn}}$, are applied to all model weights. These sub-obfuscations also protect the internal states generated during model inference. Based on the composition theorems in Section \ref{sec:composition_theorem}, we combine all the sub-obfuscations into $\phi_{\Theta}$, which together with $\phi_X, \phi_Y$ to form the covariant obfuscation of $f$ for the entire model inference process. During the online inference process, the client locally obfuscates its text data at the token level based on $\tau$ before sending them to the server. According to the definition of covariant obfuscation, \texttt{AloePri} can preserve model accuracy when performing inference on obfuscated data using obfuscated model weights.

\textbf{Remark.} \texttt{AloePri} leverages lightweight token-level permutations to protect tokens, and uses invertible transformations and noise injection to conceal the permutation itself. 
Following previous studies~\cite{du2002practical, nasr2021adversary} on practical privacy-enhancement technologies, \texttt{AloePri} aims to provide adjustable security for constrained attackers in real-world scenarios, rather than offering ideal-world security guarantees against worst-case attackers.
% These obfuscation techniques operating over an infinite domain, while inevitably introducing information leakage, have been widely studied and adopted in prior work due to their practicality \cite{}. By design, \texttt{AloePri} aims to provide adjustable security for industrial applications, rather than offering ideal-world security guarantees against worst-case attackers. 
\linyu{To estimate the security of \texttt{AloePri}, we extend the mDP \cite{10.1007/978-3-642-39077-7_5} and RDP \cite{mironov2017renyi} to derive a novel Rényi-metric Differential Privacy, which is more suitable for analyzing high-dimensional Gaussian noise, in Section \ref{sec:dp_security}. We also empirically validate the privacy guarantee against constrained attackers via a series of experiments in Section \ref{sec:exp}.
}

\subsection{Offline Model Obfuscation}
\label{model_transform}
In offline model obfuscation, token-level permutation is applied to weights of the embedding layer and model head, which enables the client to obfuscate texts at the token level and supports auto-regressive text generation tasks. Meanwhile, by introducing \textit{two-side} (left-multiplication and right-multiplication) transformations for weight matrices, the token-level permutation is kept secret from the server. We first present an algorithm to generate (inverse) key matrix used to construct obfuscations. Then we present concrete construction of $\phi^{\text{embed}}$, $\phi^{\text{head}}$, $\phi^{\text{attn}}$, $\phi^{\text{ffn}}$ for popular LLM structures.

\subsubsection{Key Matrix Generation}
During LLM inference, input data is processed layer-by-layer between adjacent layers with residual connections. \texttt{AloePri} utilizes a set of key matrices and their inverses to perform two-side transformations on the weights of adjacent layers. Therefore, without knowledge of the key matrices (or their inverses), the attacker cannot recover the secret token-level permutation imposed on the embedding layer and model head. Meanwhile, each key matrix can be canceled out by any of its inverses, ensuring the correctness of forward computation.

Algorithm \ref{alg:key_mat} shows the generation of key matrices and their inverses. The algorithm takes the hidden size $d$ of the model, expansion size $h$, and matrix coefficient $\lambda$ as inputs to initialize a set of base matrices. Then, \texttt{KeyMatGen} and \texttt{InvKeyMatGen} leverage the base matrices to generate random key matrices $\{\hat{P}\}$ and inverse key matrices $\{\hat{Q}\}$. With the expanded dimension controlled by $h$, we can generate an infinite number of matrices $\hat{P}$ and $\hat{Q}$ that satisfy $\hat{P} \cdot \hat{Q} = I$ using $\texttt{KeyMatGen}$ and $\texttt{InvKeyMatGen}$. Key and inverse key matrices are used to obfuscate model weights, including the embedding, attention, FFN, and model head layers. Since most LLMs store model weights in half-precision, an extra parameter $\lambda$ is employed to regulate the matrix norm, avoiding significantly altering the magnitude of obfuscated weights.

\begin{algorithm}[ht]
\small
\caption{Key Matrix Generation}
\label{alg:key_mat}
\begin{algorithmic}[1]

\Statex \textbf{Parameters:} hidden size $d$, expansion size $h$, coefficient $\lambda$.
\vspace{0.3em}\hrule\vspace{0.3em}

\Function{Init}{$d,h,\lambda$}
\State Uniformly sample $U$ from orthogonal group $O_d$.
\State Sample $V \sim \mathcal{N}(0,\tfrac{1}{d})^{d\times d}$, set $B = U + \lambda V$, compute $B^{-1}$.
\State Sample $E_1 \sim \mathcal{N}(0,\tfrac{1}{d})^{d\times h/2}$, $E_2 \sim \mathcal{N}(0,\tfrac{1}{d})^{h/2\times h}$, set $E = E_1E_2$.
\State Sample $F_1 \sim \mathcal{N}(0,\tfrac{1}{d})^{h\times h/2}$, $F_2 \sim \mathcal{N}(0,\tfrac{1}{d})^{h/2\times d}$, set $F = F_1F_2$.
\State Sample a orthogonal matrix $Z \in O_{d+2h}$ uniformly.
\State \Return $(B,B^{-1},E,F,Z)$.
\EndFunction

\vspace{0.15em}\hrule\vspace{0.15em}

\Function{KeyMatGen}{$B,E,F,Z$}
\State Construct $C\in\mathbb{R}^{d\times h}$ with columns sampled from $\mathrm{null}(F^T)$.
\State \Return $\hat{P} = [\,B\; C\; E\,] Z$.
\EndFunction

\vspace{0.15em}\hrule\vspace{0.15em}

\Function{InvKeyMatGen}{$B^{-1},E,F,Z$}
\State Construct $D\in\mathbb{R}^{h\times d}$ with rows sampled from $\mathrm{null}(E)$.
\State \Return $\hat{Q} = Z^T [\,B^{-1}\; F\; D\,]^T$.
\EndFunction

\end{algorithmic}
\end{algorithm}

\subsubsection{Embedding and Model Head Obfuscation}
\label{sec:token_perturb}
Since the embedding layer and model head are directly associated with the input and output of the model, we also need to take into account the design of the data obfuscation $\phi_X$ when constructing $\phi^{\text{embed}}$ and $\phi^{\text{head}}$. \texttt{AloePri} encompasses three techniques to construct $\phi^{\text{embed}}$ and $\phi^{\text{head}}$: noise addition, permutation, and (inverse) key matrix transformations.

\textbf{Noise Addition}. The client samples noise matrices $\mathcal{E}_{\text{embed}} \sim \mathcal{N}(0, \sigma^2_{e}I_{n} \otimes I_{d})$ and $\mathcal{E}_{\text{head}} \sim \mathcal{N}(0, \sigma^2_{h}I_{d} \otimes I_{n})$, where $ \sigma_e, \sigma_h$ are the standard variation of $W_{e}, W_{h}$. The noisy weights of embedding and model head can be expressed as $W^{\star}_{\text{embed}} = W_{e} + \alpha_{\text{e}} \cdot \mathcal{E}_{\text{embed}}$ and $W^{\star}_{\text{head}} =W_{h} + \alpha_{\text{h}} \cdot \mathcal{E}_{\text{head}}$, where $\alpha_{\text{e}}$ and $\alpha_{\text{h}}$ are noise parameters.

\textbf{Permutation and Key Matrix Transformation}. The client further employs token-level permutation to construct $\phi^{\text{embed}}, \phi^{\text{head}}$. This permutation keeps the support of auto-regressive generation. Meanwhile, the permutation also correlates to the construction of $\phi_X$ and $\phi_Y$ so that the attacker cannot recover private input and output data without knowing the permutation. Specifically, the client samples a permutation $\tau \sim S_{n}$. Let $\Pi$ be the permutation matrix corresponding to $\tau$. This permutation matrix is applied to $W^{\star}_{\text{embed}}, W^{\star}_{\text{head}}$. Simultaneously, a key matrix $\hat{P}_{\text{embed}}$ and a inverse key matrix $\hat{Q}_{\text{head}}$ are generated with Algorithm \ref{alg:key_mat} and used to obfuscate $W^{\star}_{\text{embed}}, W^{\star}_{\text{head}}$. Finally, the obfuscation of embedding and model head are formalized as: $\widetilde{W}_{\text{embed}} = \Pi W^{\star}_{\text{embed}} \hat{P}_{\text{embed}}$, $\widetilde{W}_{\text{head}} = \hat{Q}_{\text{head}} W^{\star}_{\text{head}} \Pi^T$.
% \begin{equation}
% \begin{cases}
%     \phi^{\text{embed}}(W_{e}) =\widetilde{W}_{\text{embed}} = \Pi W^{\star}_{\text{embed}} \hat{P}_{\text{embed}}, \\
%     \phi^{\text{head}}(W_{h}) = \widetilde{W}_{\text{head}} = \hat{Q}_{\text{head}} W^{\star}_{\text{head}} \Pi^T.
% \end{cases}
% \end{equation}
Besides the obfuscation of embedding and model head, the client uses the permutation $\tau$ to generate a secret token mapping $\mathcal{Z} = \{\mathcal{V}[i]: \mathcal{V}[\tau[i]]\}$, which is used to obfuscate data during online inference.

% We aim to ensure that key matrices do not compromise the computational correctness of the model. Given the presence of residual connections in the transformer architecture, it is necessary to consider the consistency of the obfuscation forms when designing weight obfuscation for attention and FFN layers. Specifically, the output hidden states of the attention and FFN layers are consistently obfuscated by key matrices $\{\hat{P}\}$. Meanwhile, the weights of attention, FFN layers and model head are obfuscated with $\{\hat{Q}\}$. Consequently, the hidden states can be correctly evaluated in obfuscated forms.

\subsubsection{Attention Obfuscation}
As described in Eq. \ref{eq:gqa}, an attention layer processes data using each attention head independently, and then aggregates the outputs across all heads. Therefore, we obfuscate the attention weights with the following two types of transformation.

\textbf{Intra-head Transformation.} We use Algorithm \ref{alg:gqa} to obfuscate the weights of a group of attention heads $(W^{(i)}_{\text{q}}, W^{\eta(i)}_{\text{k}}, W^{\eta(i)}_{\text{v}}, W^{(i)}_{\text{o}})$. Key matrices $\hat{Q}_{\text{q}}, \hat{Q}_{\text{k}}, \hat{Q}_{\text{v}}, \hat{P}_{\text{o}}$ (sampled via Algorithm \ref{alg:key_mat}) are applied to transform the weights. Random invertible matrices $\hat{U}_{vo}$ are used for value/output weights to preserve computation correctness. Considering RoPE, 2-dimensional rotary matrices $\hat{R}_{qk}$ and scaling matrices $\hat{H}_{qk}$ are introduced to transform query/key weights. Besides, we find that simultaneously shuffling the RoPE’s 2×2 blocks of query/key weights within a limited window exerts minimal impact on model accuracy, particularly for blocks with larger indices. Therefore, we perform block-wise permutation within a dynamic window to boost the obfuscation level.

\begin{algorithm}[ht]
\small
  \caption{Intra-head attention obfuscation}
  \label{alg:gqa}

  \begin{algorithmic}[1]
    % 初始化部分 - 无行号
    \Statex \textbf{Input:} Attention weights $(W^{(i)}_{\text{q}}, W^{\eta(i)}_{\text{k}}, W^{\eta(i)}_{\text{v}}, W^{(i)}_{\text{o}})$, maximum window size $\beta$, window sampling parameter $\gamma$, RoPE parameter $\zeta$, and block number $m_{\text{blocks}}$.
    \Statex \textbf{Output:} Obfuscated weights $(\widetilde{W}^{(i)}_{\text{q}}, \widetilde{W}^{\eta(i)}_{\text{k}}, \widetilde{W}^{\eta(i)}_{\text{v}}, \widetilde{W}^{(i)}_{\text{o}})$. 

    \vspace{0.15em}
    \hrule
    \vspace{0.15em}
    
    \State Uniformly sample $\rho_i$ from $(0,2\pi]$ and generate rotary matrix $
    \hat{R}_{\text{qk}} = \text{Diag}(\{R_i\}_{i \leq \frac{d_{head}}{2}})$, where $
    R_i = \begin{pmatrix}
        \cos\rho_i & -\sin\rho_i \\
        \sin\rho_i & \cos\rho_i
    \end{pmatrix}
    $.

    \State Sample $s_i \in \mathbb{R}$ and define
    $
    \hat{H}_{\text{qk}} = \mathrm{Diag}(\{s_1 I_2, \dots, s_{d_{\text{head}}/2} I_2\}).
    $

    \State Sample $\hat{Z}_{\text{block}} \leftarrow BlockPerm(\beta, \gamma, \zeta, m_{\text{blocks}})$.

    \State Sample $\hat{U}_{\text{vo}} \sim \mathcal{N}(0, \frac{1}{d_{\text{head}}} I_{d_{\text{head}}} \otimes I_{d_{\text{head}}})$.

    \State Sample $\hat{Q}_{\text{q}}, \hat{Q}_{\text{k}}, \hat{Q}_{\text{v}}, \hat{P}_{\text{o}}$ by Algorithm \ref{alg:key_mat}.

    \State Set
    \(
        \widetilde{W}^{\eta(i)}_{\text{k}} = \hat{Q}_{\text{k}} {W}^{\eta(i)}_{\text{k}} \hat{R}_{\text{qk}} \hat{H}^{-1}_{\text{qk}} \hat{Z}^T_{\text{block}} \), \(
        \widetilde{W}^{\eta(i)}_{\text{v}} = \hat{Q}_{\text{v}} W^{\eta(i)}_{\text{v}} \hat{U}_{\text{vo}}.
    \)
    \State Set  $\widetilde{W}^{(i)}_{\text{q}} = \hat{Q}_{\text{q}} W^{(i)}_{\text{q}} \hat{R}_{\text{qk}} \hat{H}_{\text{qk}} \hat{Z}_{\text{block}}$,  $\widetilde{W}^{(i)}_{\text{o}} = \hat{U}^{-1}_{\text{vo}} W^{(i)}_{\text{o}} \hat{P}_{\text{o}}$. \\

    \Return $(\widetilde{W}^{(i)}_{\text{q}}, \widetilde{W}^{\eta(i)}_{\text{k}}, \widetilde{W}^{\eta(i)}_{\text{v}}, \widetilde{W}^{(i)}_{\text{o}})$

    \vspace{0.3em}\hrule\vspace{0.3em}
    
    \Function{BlockPerm}{$\beta, \gamma, \zeta, m_{\text{blocks}}$}
    \State Initialize a list $\mathcal{U} = \{\}$ and set block index $t = 1$.
    \State Compute $\{\zeta_i = \zeta^{-2(i-1)/m_{\text{blocks}}} | 1 \leq i \leq m_{\text{blocks}} \}$.
    
    \While{$t < m_{\text{blocks}}$}
        \State Set $c = \min(\beta, m_{\text{blocks}} - t)$.
        \State Compute: $u = softmax(\{\zeta_{t + i} - \zeta_t | 1 \leq i \leq c\})$.
        \State Sample window size $w$ from $[1, c]$ with probabilities $u$.
        \State Uniformly sample a permutation matrix $Z \in S_{w}$ and append it to $\mathcal{U}$.
    \EndWhile
    
    \State Concatenate the matrices into a block diagonal form: $\hat{Z}_{\text{block}} = \text{BlockDiag}(\mathcal{U})$.
    
    \State Return $\hat{Z}_{\text{block}}$.

    \EndFunction
    
  \end{algorithmic}
\end{algorithm}

\textbf{Inter-head Permutation.} The attention weights are further obfuscated via attention head permutation, where we sample two random permutations: $\tau_{\text{kv}} \sim S_{m_{\text{kv}}}$ and $\tau_{\text{group}} \sim S_{m/m_{\text{kv}}}$.  $\tau_{\text{kv}}$ shuffles key and value weights at the individual attention head level, query and output weights at the grouped-head level. Meanwhile, $\tau_{\text{group}}$ shuffle query and output weights within each group, which disrupts inter-head correlations while preserving the model’s aggregation capability.
The above obfuscation method can be directly applied to MHA, MQA, and GQA. As for MLA, low-rank matrices are employed for query and key weights, with the integration of decoupled RoPE. Therefore, we obfuscate the low-rank weights in MLA using another set of invertible transformations.

\subsubsection{FFN Obfuscation}
To ensure the correctness of non-linear operations in FFN, we mainly used scaling and permutation transformation to obfuscate $\omega_{\text{ffn}} = \{W_{\text{gate}}, W_{\text{up}}, W_{\text{down}}\}$. Specifically, the obfuscated FFN weights are computed by: \(
    \widetilde{W}_{\text{gate}} = \hat{Q}_{\text{gate}} W_{\text{gate}} \hat{Z}_{\text{ffn}},  
    \widetilde{W}_{\text{up}} = \hat{Q}_{\text{up}} W_{\text{up}} \hat{H}_{\text{ffn}} \hat{Z}_{\text{ffn}}, 
    \widetilde{W}_{\text{down}} = \hat{Z}^{-1}_{\text{ffn}} \hat{H}^{-1}_{\text{ffn}} W_{\text{down}} \hat{P}_{\text{down}}, 
\)
where $\hat{Z}_{\text{ffn}}$ is a permutation matrix uniformly sampled from $S_{d_{\text{ffn}}}$, and $\hat{H}_{\text{ffn}}$ is a randomly sampled scaling matrix.

For MoE models, an additional router weight \( W_{\text{router}} \) is incorporated to select experts. To obfuscate \( W_{\text{router}} \), we first normalize the weight vector corresponding to every expert to get $W_{\text{router}}'$. Meanwhile, we sample a \( m_{\text{exp}} \times m_{\text{exp}} \)-dimensional permutation matrix \( \hat{Z}_{\text{router}} \) to generate \( \widetilde{W}_{\text{router}} = \hat{Q}_{\text{router}} W_{\text{router}}' \hat{Z}_{\text{router}} \). The order of experts is shuffled according to \( \hat{Z}_{\text{router}} \), ensuring that the experts can be selected correctly.

\subsubsection{Layer Normalization Transformation}
% As RMSNorm leverages a one-dimensional weight $\bm{w}_{\text{norm}}$ to normalize hidden layer outputs, we also need to adjust this weight to ensure the correctness of the calculation. 

Let \( W_{\text{norm}} = \text{Diag}(w_{\text{norm}}) \) denote the diagonal matrix corresponding to the RMSNorm weights. Assuming that the input data $x$ of any normalization layer follows a Gaussian distribution, we use $ \kappa = \mathbb{E}[\frac{||x\hat{P}||}{||x||}] $ as the coefficient for the obfuscated normalization layer to adjust for the bias induced by the $\hat{P}$ transformation. We then fuse an RMSNorm layer with weights \( \widetilde{w}_{\text{norm}} = \bm{1} \cdot \kappa \) and a linear layer with weights \( W_{\text{norm}} \) to replace the plaintext RMSNorm layer. The weights of the linear layer \( W_{\text{norm}} \) can be merged into the layer adjacent to the RMSNorm layer before applying weight obfuscation.

\subsection{Online Inference}
\label{sec:aloepri_online}
% Recall that in Section \ref{sec:token_perturb}, a secret token mapping $\mathcal{Z}$ is stored on the client side. 
In the online inference phase, the client leverages the secret token mapping $\mathcal{Z}$ to obfuscate input prompts. Specifically, the client first tokenizes a prompt locally into $l$ tokens $\mathcal{T} = \{ tok_1, \cdots, tok_l \}$, then encodes each token $tok_i$ into its obfuscated counterpart $\widetilde{tok}_i = \mathcal{Z}[tok_i]$. Once all obfuscated tokens are generated, the client detokenizes them to construct an obfuscated prompt $\widetilde{\mathcal{T}}$ and sends it to the server. The server tokenizes $\widetilde{\mathcal{T}}$ to produce obfuscated token indices $\{ \widetilde{x}_i \}_{i \leq l}$ as the input of the obfuscated model. The model then outputs an obfuscated response to the client. Finally, the client utilizes $\mathcal{Z}$ to decode the obfuscated response back to plaintext.

\linyu{\textbf{Remark.} Leveraging covariant obfuscation, which is capable of performing obfuscation in the data space, \texttt{AloePri} is compatible with data obfuscation techniques operating on the token space for LLM inference. For example, similar to RANTEXT \cite{tong2025inferdpt}, SANTEXT \cite{yue2021differential}, and CUSTEXT \cite{chen2023customized}, \texttt{AloePri} can perform token perturbation based on the similarity of token embeddings before applying secret token mapping during online inference.}

\subsection{\wqruan{Accuracy} Analysis}
We first analyze the obfuscation-induced bias of each LLM component individually, then leverage the composition theorems for covariant obfuscation to show that the inference accuracy of LLMs can be preserved. We adopt the symbol $\approx_{e_{C}}$ to denote the approximation relation with the obfuscation error $e_{C}$ as a constant upper bound.
% Note that we still use $\widetilde{\square}$ to note the obfuscated data and weights transformed by obfuscations $\phi_X, \phi_\Theta$.

\textbf{Embedding and Model Head.} Let $x$ denote an input token index, and $f_{\text{embed}}, \widetilde{f}_{\text{embed}}$ the embedding functions in the plaintext and obfuscated spaces, respectively. The covariant obfuscation of the embedding layer is formalized as: 
\(
    \phi^{\text{embed}}_X(x) = \tau(x), \phi^{\text{embed}}_\Theta(W_{e}) = \Pi W^{\star}_{\text{embed}}\hat{P}_{\text{embed}}, \phi^{\text{embed}}_Y(y) = y \hat{P}_{\text{embed}}, \psi^{\text{embed}}_Y(\widetilde{y}) = \widetilde{y} \hat{Q}.
\)
Accordingly, $\widetilde{f}_{\text{embed}}(\phi^{\text{embed}}_X(x), \phi^{\text{embed}}_\Theta(W_{e})) \approx_{e^{\text{embed}}_{C}} \phi^{\text{embed}}_Y\circ f_{\text{embed}}(x, W_{e})$, where $e^{\text{embed}}_{C}$ is the obfuscation error of this covariant obfuscation. Similarly, we derive $\widetilde{f}_{\text{head}}(\phi^{\text{head}}_X(x), \phi^{\text{head}}_\Theta(W_{h})) \approx_{e^{\text{head}}_{C}} \phi^{\text{head}}_Y\circ f_{\text{head}}(x, W_{h})$, with $e^{\text{head}}_{C}$ as the corresponding obfuscation error.

% \begin{align*}
%     f_{\text{embed}}(\widetilde{x_i}, \widetilde{W}_{\text{embed}}) &=\widetilde{W}_{\text{embed}}[\widetilde{x_i}] \\ &\approx W_{e}[x_i] \hat{P}_{\text{embed}} \\ 
%     &= f(x_i, W_{e}) \hat{P}_{\text{embed}}
% \end{align*}

\textbf{Attention.} We denote $\phi^{\text{attn}}_X(x) = x \hat{P}$ as the obfuscated input to the attention layer $\widetilde{\omega}_{\text{attn}}$. The attention score for the $i$-th head is computed as:
\(
\mathcal{G}(\widetilde{x} \widetilde{W}^{(i)}_{\text{q}}) \mathcal{G}(\widetilde{x}  \widetilde{W}^{\eta(i)}_{\text{k}})^T 
\approx_{e^{\text{attn}}_C}\mathcal{G}(x W^{(i')}_{\text{q}}) \mathcal{G}(x W^{\eta(i')}_{\text{k}} )^T,
\)
where $i'$ denotes the permuted index corresponding to the $i$-th head, and $e^{\text{attn}}_C$ is the obfuscation error induced by block permutation (parameterized by $\beta, \gamma$). Meanwhile, the attention value satisfies:
\(
\widetilde{x} \widetilde{W}^{\eta(i)}_{\text{v}} \widetilde{W}^{(i)}_{\text{o}} = x W^{\eta(i')}_{\text{v}} W^{(i')}_{\text{o}} \hat{P}_{\text{o}}.
\)
Thus, by integrating attention scores and values, we construct the covariant obfuscation as follows: $\phi^{\text{attn}}_X(x) = x \hat{P}$, $\phi^{\text{attn}}_\Theta(\omega_{\text{attn}}) = \widetilde{\omega}_{\text{attn}}$, $\phi^{\text{attn}}_Y(y) = y \hat{P}_{\text{o}}$, $\psi^{\text{attn}}_Y(\widetilde{y}) = \widetilde{y} \hat{Q}$, with $\widetilde{f}_{\text{attn}}$ denoting the attention function over obfuscated spaces. This obfuscation satisfies:
\(
\widetilde{f}_{\text{attn}}(\phi^{\text{attn}}_X(x), \phi^{\text{attn}}_\Theta(\omega_{\text{attn}})) \approx_{e^{\text{attn}}_{C}} \phi^{\text{attn}}_Y \circ f_{\text{attn}}(x, \omega_{\text{attn}}),
\)
where $f_{\text{attn}}(x, \omega_{\text{attn}}) \hat{P}_{\text{o}} = \phi^{\text{attn}}_Y \circ f_{\text{attn}}(x, \omega_{\text{attn}})$.

\textbf{FFN.} Taking an $\widetilde{x} = x \hat{P}$ as input, the forward computation of the obfuscated FFN layer holds that:
\(
 \widetilde{f}_{\text{ffn}}(\widetilde{x}, \widetilde{\omega}_{\text{ffn}}) 
%  = \left(
% \text{SiLU}(\widetilde{x}\widetilde{W}_{gate}) \odot (\widetilde{x}\widetilde{W}_{up})
% \right) \cdot \widetilde{W}_{down} \\
% = \left(
% \text{SiLU}(x \hat{P} \hat{Q}_{\text{gate}}W_{\text{gate}} \hat{Z}_{\text{ffn}}) \odot (x \hat{P} \hat{Q}_{\text{up}}W_{\text{up}} \hat{H}_{\text{ffn}}\hat{Z}_{\text{ffn}}) 
% \right) \cdot \hat{Z}^{-1}_{\text{ffn}} \hat{H}^{-1}_{\text{ffn}} W_{\text{down}} \hat{P}_{\text{down}} 
= \left(
\text{SiLU}(x W_{\text{gate}}) \odot (x W_{\text{up}}) 
\right) W_{\text{down}} \hat{P}_{\text{down}} 
=  f_{\text{ffn}}(x, \omega_{\text{ffn}}) \hat{P}_{\text{down}} 
= \phi^{\text{ffn}}_Y \circ f_{\text{ffn}}(x, \omega_{\text{ffn}})
\), 
where $\phi^{\text{ffn}}_Y(y) = y \hat{P}_{\text{down}}$ and $\widetilde{f}_{\text{ffn}}$ is the FFN function defined over the obfuscated spaces.

\textbf{Layer Normalization.} Let $W_{\text{norm}} = \text{Diag}(w_{\text{norm}})$, we construct the covariant obfuscation as: $\phi^{\text{norm}}_X(x) = x \hat{P}$, $\phi^{\text{norm}}_\Theta(w_{\text{norm}}) = \widetilde{w}_{\text{norm}}$, $\phi^{\text{norm}}_Y(y) = y W^{-1}_{\text{norm}} \hat{P}$, $\psi^{\text{norm}}_Y(\widetilde{y}) = \widetilde{y} \hat{Q} W_{\text{norm}}$, and $\widetilde{f}_\text{RMSNorm}$ denotes the RMSNorm function defined over the obfuscated spaces. We have 
\(
     \widetilde{f}_\text{RMSNorm}(\widetilde{x}, \widetilde{w}_{\text{norm}}) 
    =  \frac{\bm{x} \hat{P} }{\sqrt{ \frac{1}{d+2h}\sum^{d + 2h}_{i} \widetilde{x}^2_i}} \kappa I   
    \approx_{e^{\text{norm}}_{C}}   \frac{\bm{x} \hat{P} }{\sqrt{ \frac{1}{d}\sum^d_{i} x^2_i}}    
    % \approx & \frac{\bm{x} }{\sqrt{ \frac{1}{d}\sum^d_{i} x^2_i}} (\text{Diag}(w_{\text{norm}})W) \\
    % & =  f_\text{RMSNorm}(x, w_{\text{norm}}) W^{-1}_{\text{norm}}\hat{P} \\
    =\phi^{\text{norm}}_Y \circ f_\text{RMSNorm}(x, w_{\text{norm}})
\),
where $e^{\text{norm}}_{C}$ is the obfuscation error of the covariant obfuscation. 
% In the first approximation, we use $\kappa$ to convert the root mean square of the obfuscated input to that of the plaintext input. While this operation introduces a small error, we find in subsequent experiments that the error has a negligible impact on the model accuracy. Note that $W_{\text{norm}}$ is merged into adjacent layers before weight obfuscation.

% \textbf{Model Head.} The covariant obfuscation of model head is formalized as: $\phi^{\text{head}}_X(x) = x \hat{P}$, $\phi^{\text{head}}_\Theta(W_{h}) = \hat{Q}_{\text{head}} W^{\star}_{\text{head}} \Pi^T$, $\phi^{\text{head}}_Y(y) = y \Pi^T$, $\psi^{\text{head}}_Y(\widetilde{y}) = \widetilde{y} \Pi$, and $\widetilde{f}_\text{head}$ denotes the model head function defined over the obfuscated spaces. Therefore, it holds that \(\widetilde{f}_{\text{head}}(\phi^{\text{head}}_X(x), \phi^{\text{head}}_\Theta(W_{h})) \approx_{e^{\text{head}}_{C}} \phi^{\text{head}}_Y\circ f_{\text{head}}(x, W_{h})\), where $e^{\text{head}}_{C}$ is the obfuscation error for the covariant obfuscation. 

\textbf{Putting Together.} The above components are sequentially connected with residue connections in typical LLM structures, and the error of each component is relatively small. Since the components satisfy the bound conditions of Theorem \ref{the:sequantial} and Theorem \ref{the:summation}, we can integrate these components to construct the covariant obfuscation as: $\phi_X(x) = \tau(x)$, $\phi_\Theta(\theta) = (\widetilde{W}_{\text{embed}}, \widetilde{W}_{\text{head}}, \{\widetilde{\omega}_{\text{attn}}, \widetilde{\omega}_{\text{ffn}}, \widetilde{w}_{\text{norm}}\})$, $\phi_Y(y) = \tau(y)$, $\psi_Y(\widetilde{y}) = \tau^{-1}(\widetilde{y})$, and $\widetilde{f}$ is the LLM inference function for obfuscated spaces. It holds that
$\widetilde{f}_\Theta(\phi_X(x), \phi_\Theta(\theta)) \approx_{e^{\texttt{AloePri}}_{C}} \phi_Y\circ f_\Theta(x, \theta)$, where $e^{\texttt{AloePri}}_{C}$ is the obfuscation error of \texttt{AloePri}. 

As $e^{\texttt{AloePri}}_{C}$ is related to the structure of LLM, we take a typical dense model structure like Qwen2 \cite{qwen2.5} as an example. The obfuscation error of $i$-th decoder layer can be derived as 
\[
e^{\text{decoder}}_{C_i} \leq (M^{\text{norm}}_i(M^{\text{attn}}_i e^{\text{norm}}_{C_i} + e^{\text{attn}}_{C_i}) + e^{\text{norm}}_{C_i})M^{\text{norm}}_i M^{\text{FFN}}_i,
\]
where $M^{\text{norm}}, M^{\text{attn}}_i, M^{\text{FFN}}_i$ are the Lipschitz constants of $i$-th RMSNorm, attention, FFN layers with consideration of residue connection. Let $M^{\text{head}}$ and $M^{\text{decoder}}_i$ denote the Lipschitz constants of model head and $i$-th decoder layer, respectively.
It holds that $e^{\texttt{AloePri}}_{C} \leq \mathcal{M}_0 e^{\text{embed}}_C + \sum _{i=1}^{L} \mathcal{M}_i e^{\text{decoder}}_{C_i} + e^{\text{head}}_{C}$, where $\mathcal{M}_i=M^{\text{head}}\prod _{j=i+1}^{L} M^{\text{decoder}}_i$.

% Moreover, 
% let $e_0$ be the obfuscation
% error of embedding layer, 
% $e_i, M_i$ be the  obfuscation
% error and  Lipschitz constant of $i$-th transformer block  respectively for $i=1, \cdots, L$, 
% $e_{L+1}, M_{L+1}$ be the obfuscation
% error and Lipschitz constant of Model Head respectively.
% Then we have the obfuscation
% error of AloePri can bounded by
% \begin{align*}
% e_{\mathsf{AloePri}}\leq  \sum _{i=0}^{L+1} \mathcal{M}_ie_{i}
% \end{align*}
% where $\mathcal{M}_i=\prod _{j=i+1}^{L+1} M_i$ \qed

% \input{chap/5.analysis}

\section{Security Analysis}
\label{sec:dp_security}
\linyu{In this section, we introduce Rényi-metric Differential Privacy (RmDP). Building on RmDP, we first analyze the privacy budget of the secret permutation $\tau$ associated with data obfuscation. We then compare the privacy guarantees of sensitive tokens provided by \texttt{AloePri} and other data-only obfuscation methods.}

\subsection{R\'enyi-metric Differential Privacy}
% \qizhi{
\begin{Qizhi}

\linyu{DP is a natural framework for analyzing privacy protection methods, but it is not directly applicable to our setting. On the one hand, the privacy goal in LLM inference is to protect clients’ private texts, which differs fundamentally from the notion of adjacent datasets in DP. On the other hand, the high-dimensional Gaussian noise commonly used in this setting is difficult to characterize under standard DP. To address these two issues in a unified way, we introduce RmDP in Definition \ref{def:rmdp}.}

% In general, the concept of differential privacy is used to analyze the security of perturbation mechanisms. However, traditional differential privacy is inapplicable to our problem, mainly for two reasons:

% \begin{enumerate}
%     \item Traditional differential privacy only focuses on \emph{adjacent datasets}, but we care about the privacy of \emph{obfuscated keys}.
%     \item Traditional differential privacy excels at handling 1-dimensional Laplace noise, but is not well-suited for high-dimensional Gaussian noise (to handle Gaussian noise, one must extend $\epsilon$-DP to $\epsilon$-$\delta$ DP, whose composition theorems become far more complex).
% \end{enumerate}

% To clarify these two obstacles, we observe the following:

% \begin{itemize}
%     \item A metric can be defined on the permutation group $\mathbb{Z}_n^l$ where the obfuscated keys reside, forming a metric space. This allows us to use \emph{metric differential privacy} \cite{xie2025decademetricdifferentialprivacy}  to characterize the protection strength of obfuscated keys on $S_n$. However, it remains difficult to handle Gaussian noise.
%     \item \emph{R\'enyi differential privacy} \cite{Mironov_2017} can handle high-dimensional Gaussian noise, but its privacy protection metric is limited to \emph{adjacent datasets}.
% \end{itemize}

\begin{definition}[R\'enyi-metric Differential Privacy]
\label{def:rmdp}
Let $(\mathcal{X}, d)$ be a metric space. A randomized mechanism $\mathcal{M}\colon \mathcal{X} \to \mathcal{Y}$ satisfies \emph{$(\alpha,\epsilon,d)$-RmDP} if, for any $x,x' \in \mathcal{X}$,
\[
D_\alpha\!\left(p_{\mathcal{M}(x)} \,\|\, p_{\mathcal{M}(x')}\right)
\le \epsilon \, d(x, x'),
\]
where $\alpha>1$, $\epsilon \ge 0$, and $p_{\mathcal{M}(x)}$ denotes the distribution of the mechanism output on input $x$.
\end{definition}

For LLM inference, RmDP offers two key advantages:
\begin{itemize}
    \item RmDP enables quantifying the privacy protection of clients' sensitive tokens by defining a metric on $\mathbb{Z}_n^l$.
    \item By leveraging the $\alpha$-R\'enyi divergence, RmDP can handle high-dimensional Gaussian noise without resorting to a relaxed variant.
\end{itemize}
Table \ref{tab:compare_dp} compares DP, R\'enyi DP, metric DP, and R\'enyi-metric DP.

\begin{table}[h]
\centering
\caption{Comparison of DP, R\'enyi DP, metric DP, and R\'enyi-metric DP.}
\label{tab:compare_dp}
\begin{tabular}{lcccc}
\toprule
 & DP & R\'enyi DP & Metric DP & R\'enyi-metric DP \\
\midrule
Privacy notion & Adjacent datasets & Adjacent datasets & Metric space & Metric space \\
Gaussian mechanism analysis & Not natural & Natural & Not natural & Natural \\
\bottomrule
\end{tabular}
\end{table}

\subsection{Security Analysis of Private Tokens}
To characterize the privacy guarantee on token-sequence space $\mathbb{Z}_n^l$ with RmDP, we define a metric in Definition \ref{def:zn_perm_metric}.

\begin{definition}[Permutation Metric on $\mathbb{Z}^{l}_{n}$]
\label{def:zn_perm_metric}
For any two token sequences $x, x' \in \mathbb{Z}_n^l$, define $d(x,x')$ as the minimum number of transpositions needed to transform $x$ into $x'$. More precisely, $d(x,x')$ is the smallest integer $k$ such that there exists a sequence
\[
x = x_0 \rightarrow x_1 \rightarrow \cdots \rightarrow x_k = x'
\]
where, for each $i=0,\ldots,k-1$, there exists a transposition (a permutation that swaps exactly two elements) $g_i \in S_n$ satisfying
\(x_{i+1} = g_i x_i\).
\end{definition}

As stated in Section \ref{sec:aloepri_online}, \texttt{AloePri} is compatible with data obfuscation techniques based on token perturbation. We consider a standard exponential mechanism for token perturbation, denoted by $\mathcal{M}_1$, satisfying 
\[
p_{\mathcal{M}_1(x)}(y) \propto e^{-\epsilon _1 d(x,y)}.
\]
We then derive the privacy guarantee of \texttt{AloePri} in Theorem \ref{the:aloepri_rmdp} (full proof in Appendix \ref{appx:proof_rmdp}) when applying $\mathcal{M}_1$ in the online phase.

\begin{theorem} 
\label{the:aloepri_rmdp}
Let \(d(\cdot, \cdot) \) be the metric in Definition \ref{def:zn_perm_metric}. \texttt{AloePri} satisfies (\(\alpha\)-\(\epsilon\)-\(d\))-RmDP, where \(\alpha = 2\), and
\begin{equation*}
    \epsilon= \left \{ 
    \begin{array}{cc}\epsilon_1 - \frac{\epsilon _1 ^2 }{4(n-1)\epsilon _2},  &    \mathsf{ if }  \epsilon _1 \leq 2(n-1) \epsilon _2 \\
    (n-1) \epsilon _2,  & \mathsf{ otherwise}
    \end{array}
    \right .
\end{equation*}
Here
\[
\epsilon_2 = \pi^2 \cdot (\epsilon_e + \epsilon_h),
\epsilon_e = \frac{\alpha\left(\lambda_1^2(W_e) + \lambda_2^2(W_e)\right)}{4\sigma_e^2},
\epsilon_h = \frac{\alpha\left(\lambda_1^2(W_h) + \lambda_2^2(W_h)\right)}{4\sigma_h^2},
\]
where \(\lambda_1(W_e) \ge \lambda_2(W_e)\) are the largest and second-largest singular values of \(W_e\), and \(\lambda_1(W_h) \ge \lambda_2(W_h)\) are the largest and second-largest singular values of \(W_h\).
\medskip
\end{theorem}

\begin{remark} \texttt{AloePri}'s consumed privacy budget \(\epsilon\) is strictly smaller than the consumed privacy budget \(\epsilon_1\) of data obfuscation $\mathcal{M}_1$.
\end{remark}

Theorem \ref{the:aloepri_rmdp} demonstrates that the privacy protection of \texttt{AloePri} exhibits two key properties.
When the data obfuscation mechanism \(\mathcal{M}_1\) provides satisfactory privacy guarantees (i.e., \(\epsilon_1 \leq (n-1)\epsilon_2\)), \texttt{AloePri} can further reduce the privacy budget. When the privacy protection offered by \(\mathcal{M}_1\) is limited, \texttt{AloePri} still ensures smaller privacy budget consumption through model obfuscation.

\end{Qizhi}

\section{Experiments}

In this section, we present the experimental results of \texttt{AloePri} and their explanation.

\label{sec:exp}
\subsection{Experimental Settings}

\textbf{Models and Datasets.} We conduct experiments on commonly-used open-source LLMs to evaluate the effectiveness of \texttt{AloePri}. For dense models, we select Qwen2.5/Qwen3 \cite{qwen2.5,qwen3technicalreport}, Llama3 \cite{touvron2023llama}, Deepseek-R1-Distill-Qwen (R1-Distill) \cite{deepseekai2025deepseekr1incentivizingreasoningcapability}. For MoE models, we choose Qwen3-MoE and Deepseek-V3.1-Terminus \cite{deepseekai2024deepseekv3technicalreport}, covering both moderate and large-scale models.

% \linyu{We evaluate the model utility using the OpenCompass framework \cite{2023opencompass}. 
We employ various representative LLM evaluation benchmarks to evaluate model accuracy, including SST2 \cite{socher-etal-2013-recursive} for sentence classification, MMLU \cite{hendryckstest2021} and C-Eval \cite{huang2023ceval} for general tasks, HumanEval \cite{chen2021evaluating} for code generation, IFEval \cite{zhou2023instructionfollowingevaluationlargelanguage} for alignment, and PIQA \cite{Bisk2020} for physical commonsense reasoning. Additionally, we use PUPA dataset \cite{siyan2024papillon} to validate the protection level for Personally Identifiable Information (PII). We use CCI3 \cite{wang2024cci30}, Huatuo26M-Lite \cite{li2023huatuo26m}, and MedDialog \cite{chen2020meddiag} datasets to evaluate information leakage related to token frequency patterns.

\noindent\textbf{Baseline Methods.} According to the discussion in Section \ref{sec:related}, we compare \texttt{AloePri} with four obfuscation-based privacy-preserving inference methods, which are the most related to our work. \texttt{SANTEXT} \cite{yue2021differential} and \texttt{RANTEXT} \cite{tong2025inferdpt} are based on token substitution, while \texttt{DP-Forward} \cite{du2023dp} and \texttt{SGT} \cite{roberts2025learning} are based on embedding transformation. \wqruan{Because cryptography-based methods are still far from practical usage and TEE-based methods lack hardware compatibility, we do not compare \texttt{AloePri} with them.} 
% Regarding STIP \cite{yuan2023secure}, it achieves private inference with near-plaintext accuracy through a distinct threat model where the plaintext model weights are kept secret from the server. This setup significantly reduces the attacker’s capabilities compared to our threat model while also limiting its applicable scenarios. For fairness, we do not include STIP in our experiments.

\noindent\textbf{Attacks.} We evaluate the resistance of \texttt{AloePri} against three types of typical attacks. (1) Obfuscation recovery: We adopt Vocabulary-Matching Attack (VMA) \cite{thomashidden} and Invariant Attack (IA) \cite{lin2024inversion} to verify whether the attacker can recover the secret mapping of \texttt{AloePri} based on the relationship between plaintext and obfuscated mode weights. (2) Training-based inversion: We use Internal State Attack (ISA) \cite{dong2025depth}, Inversion Model Attack (IMA) \cite{kugler2021invbert}, and nearest neighbor (NN) attack to investigate the protection of internal states. (3) Token frequency exploit: We use Token Frequency Matching Attack (TFMA) and Substitution Deciphering Attack (SDA) \cite{aldarrab-may-2021-sequence} to test the information leakage of token frequency. To the best of our knowledge, these attacks encompass almost all privacy threats faced by obfuscation-based privacy-preserving inference methods. We provide detailed information about the above attacks in Appendix \ref{appx:attack}.

\noindent\textbf{Privacy Metrics.} We adopt several metrics to quantify the privacy level. The \textit{Text Token Recovery Success Ratio (TTRSR)} measures the proportion of original tokens successfully recovered by the attacker, providing a direct measure of token-level privacy protection. The \textit{PII Recovery Success Ratio (PIIRSR)} quantifies the percentage of personally identifiable information (e.g., names, addresses, phone numbers) correctly extracted from the obfuscated outputs. The \textit{Recovered Text Cosine Similarity (CosSim)} calculates the cosine similarity between the embedding of the recovered text and the original plaintext. \textit{Top-k} accuracy assesses the accuracy of token matching based on the k-nearest neighbors. The \textit{BLEU-4} score measures the n-gram overlap between the adversary's recovered text and the original plaintext.

\noindent\textbf{Hyperparameter Configuration.}
In accordance with the recommended configurations of these models, we fix the generation parameters across all experiments with temperature = 0.65, top-k = 20, and top-p = 0.95. As for the privacy parameters, we tune them in subsequent experiments. If not specified, we set the matrix coefficient to $\lambda=0.3$, the expansion size to $h=128$, the noise coefficients to $\alpha_{e} = 1.0$ and $\alpha_{h} =0.2$, and the attention block-wise parameters to $\beta=8$ and $\gamma=1\mathrm{e}^3$ by default. We present detailed hyperparameter settings of subsequent experiments in Appendix \ref{appx:hyperparam}.

\begin{table*}[!htbp]
\centering
\small  
\caption{\qizhi{Accuracy and privacy comparison with baselines on Qwen2.5-14B-Instruct.}}
\label{tab:baseline_comp}
\begin{threeparttable}
\scalebox{0.8}{
\begin{tabular}{>{\centering\arraybackslash}p{2.5cm}|>
{\centering\arraybackslash}p{2cm}|>{\centering\arraybackslash}p{1.5cm}>{\centering\arraybackslash}p{1.5cm}>{\centering\arraybackslash}p{1.5cm}|>{\centering\arraybackslash}p{1.5cm}>{\centering\arraybackslash}p{1.5cm}>{\centering\arraybackslash}p{1.5cm}}
\toprule \toprule
\multirow{2}{*}{Method} & \multicolumn{1}{c|}{Classification Acc. $\uparrow$} & \multicolumn{3}{c|}{Generation Acc. $\uparrow$}  & \multicolumn{3}{c}{Privacy\tnote{1}}   \\
& SST2 & MMLU & PIQA & IFEval & Attack & TTRSR(\%) $\downarrow$ & CosSim $\downarrow$ \\
\midrule
\texttt{Plaintext} & 97.25 & 81.95 & 77.45 & 77.09 & - & - & -  \\ \midrule
\texttt{SANTEXT}  & 81.31 & 51.07 & 49.72 & 57.45 & (Direct obs.) & 62.90 & 0.93  \\ \midrule
\texttt{RANTEXT}  & 87.50 & 37.77 & 18.30 & 57.18 & (Direct obs.) & 56.92 & 0.88  \\ \midrule
\multirow{2}{*}{\texttt{DP-Forward}\tnote{2}}  & \multirow{2}{*}{93.60} & \multirow{2}{*}{-}  & \multirow{2}{*}{-} & \multirow{2}{*}{-} & NN  & 37.82 & 0.51 \\ 
& &  &  & & IMA  & 77.08 & 0.84   \\ \midrule
\multirow{2}{*}{\texttt{SGT}\tnote{3}} & \multirow{2}{*}{92.43} & \multirow{2}{*}{30.30} & \multirow{2}{*}{26.06} & \multirow{2}{*}{34.66} 
& NN & 11.12 & 0.69   \\
& & & & & IMA & 98.58 & 0.99   \\ \midrule

\multirow{5}{*}{\texttt{\texttt{AloePri}}} & \multirow{5}{*}{\textbf{97.13}}  & \multirow{5}{*}{\textbf{80.61}}   & \multirow{5}{*}{\textbf{75.23}} & \multirow{5}{*}{\textbf{79.49}} & VMA & 13.51 & 0.31  \\
& & & & & NN & 0.0 & 0.42   \\
& & & & & IMA & 0.0 & 0.36   \\
& & & & & IA & 5.95 & 0.40 \\
& & & & & ISA & 0.0 & 0.20   \\

\bottomrule \bottomrule
\end{tabular}
}
\begin{tablenotes}
    \item[1] Only count the privacy of the input text, as the baselines do not protect generated texts.
    \item[2] \texttt{DP-Forward} cannot support for auto-regressive generation tasks.
    \item[3] \texttt{SGT} does not release source codes. Therefore, we re-implement \texttt{SGT} according to the description in their paper~\cite{roberts2025learning}.
\end{tablenotes}
        
\end{threeparttable}
\end{table*}

\noindent\textbf{Experimental Environment.}
% To simulate real-world deployment scenarios, we establish distinct computational environments for the client and server components. 
The client-side environment is configured with a CPU-only setup equipped with two Intel(R) Xeon(R) Platinum 8457C processors (96 cores). The server-side environment employs a high-performance GPU cluster to handle the computational demands of large language model inference. We deploy the server-side inference service with vLLM 0.9.1 \cite{kwon2023efficient} and CUDA 12.4.

\subsection{Comparison with Previous Methods}
We compared \texttt{AloePri} with data obfuscation methods on model accuracy and privacy. The model accuracy is evaluated on MMLU, IFEval, and PIQA. We conducted various attacks (e.g. NN, IMA, VMA, IA, and ISA) to test privacy levels of different methods. We present the details and hyperparameter settings of baselines in Appendix \ref{appx:imple_details}. 
% The hyperparameters are configured to ensure their comparison under comparable utility or privacy levels. 

As shown in Table \ref{tab:baseline_comp}, \texttt{AloePri} significantly outperforms baselines in both model accuracy and data privacy. For \texttt{SANTEXT} and \texttt{RANTEXT}, we calculated the TTRSR by counting the number of plaintext tokens directly observed by the attacker. The two methods leak over 50\% of tokens while still suffer significant accuracy degradation. The reason is that the sensitive tokens replacement directly disrupts the key semantic information needed for answer generation. 
% Among embedding-level obfuscation methods, \texttt{SGT} outperforms \texttt{DP-Forward} in both model accuracy and privacy protection. 
% For \texttt{DP-Forward}, the embedding layer and the first decoder layer (nearly 1B parameters) are offloaded to the client side before noise addition to achieve better accuracy and privacy. Nevertheless, 
\texttt{DP-Forward} reaches about 94\% accuracy on SST2, but it is vulnerable to IMA with over 75\% TTRSR. In \texttt{SGT}, the client leverages a 1.6B-parameter model for noise generation. However, \texttt{SGT} still causes over 50\% accuracy loss. Meanwhile, \texttt{SGT} is also vulnerable to IMA, leading to over 90\% TTRSR. Based on covariant obfuscation, \texttt{AloePri} achieves a better balance between model accuracy and privacy. \texttt{AloePri} shows strong resistance to VMA, IMA, IA, and ISA, i.e., TTRSR is less than 15\%. Meanwhile, \texttt{AloePri} achieves less than 3\% accuracy loss on both classification and generation tasks.

% \qizhi{
% Additionally, we used the Contrastive Log-ratio Upper Bound (CLUB) \cite{cheng2020club} estimator (details presented in Appendix \ref{appx:mi_simulation}) to simulate the MI between pair-wise plaintext and obfuscated token, whose theoretical upper bound is $\log(n)$. Through MI, we can theoretically compare the information leakage levels of different methods. Due to the advantage of covariant obfuscation, \texttt{AloePri} reaches the smallest information leakage. For DP-Forward, excessive information leakage and the limitations of CLUB result in the simulated MI exceeding the upper bound $\log(n)$. It is worth noting that MI is not always correlated with the success rate of specific attacks. This is because each specific attack has inherent limitations and may fail to capture the leaked information. 
% }

\begin{figure}[ht]
    \centering
    % 第一个子图
    \begin{subfigure}[b]{0.4\textwidth}
        \centering
        \includegraphics[width=\textwidth]{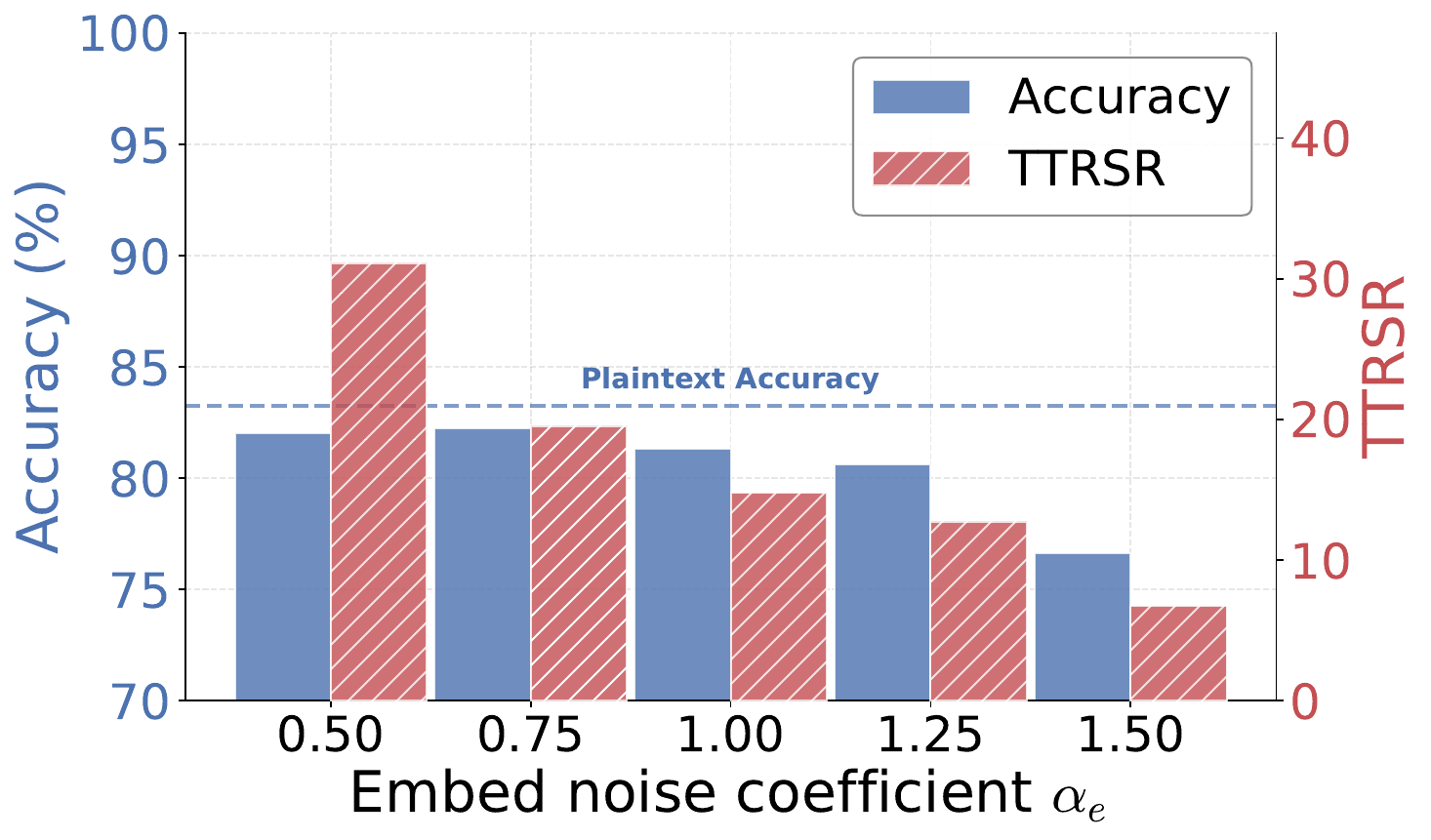}
        \caption{\small Various $\alpha_{\text{e}}$ with fixed $\alpha_{\text{h}} = 0.2$} 
        \label{fig:privacy-utility-1}
    \end{subfigure}
    % 第二个子图
    \begin{subfigure}[b]{0.4\textwidth}
        \centering
        \includegraphics[width=\textwidth]{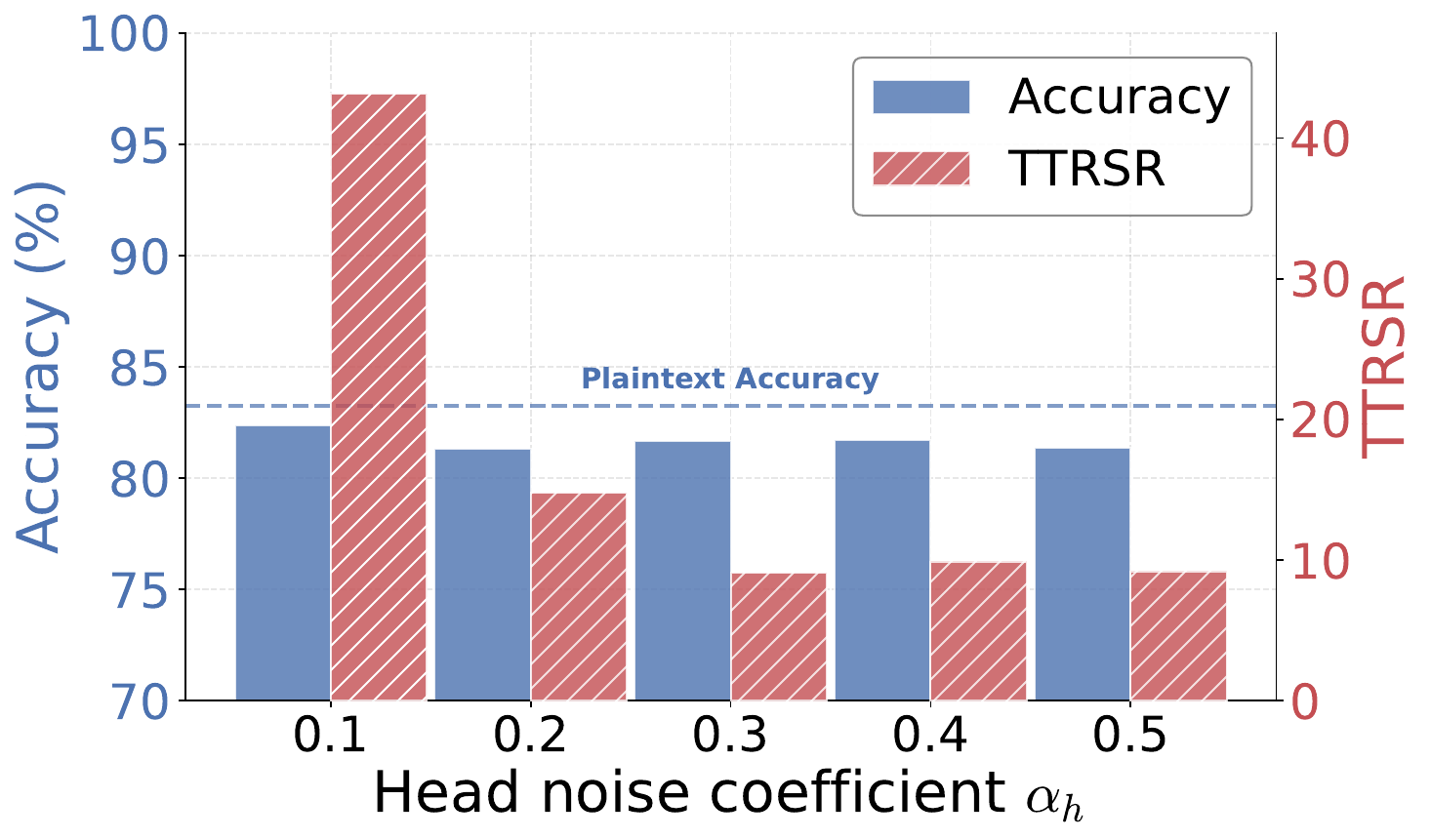}
        \caption{\small Various $\alpha_{\text{h}}$ with fixed $\alpha_{\text{e}} = 1.0$ }
        \label{fig:privacy-utility-2}
    \end{subfigure}
    \caption{Privacy (TTRSR) and accuracy under various noise parameters on Qwen2.5-14B-Instruct and C-Eval.}
    \label{fig:privacy-utility}
\end{figure}

\begin{figure}[ht]
    \centering
    % 第一个子图
    \begin{subfigure}[b]{0.4\textwidth}
        \centering
        \includegraphics[width=\textwidth]{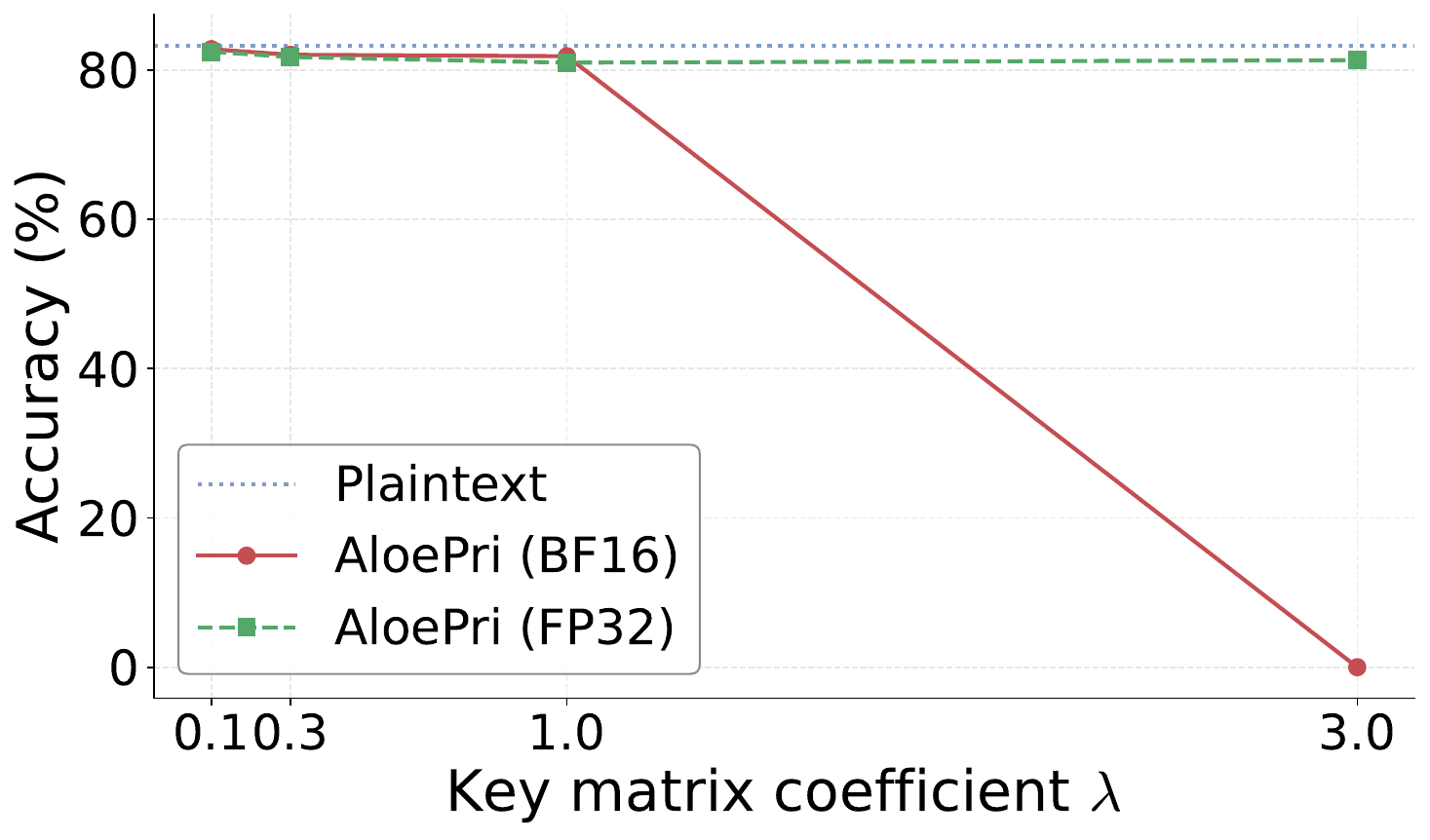}
        \caption{\small Accuracy of C-Eval.} 
        \label{fig:lambda_acc}
    \end{subfigure}
    % 第二个子图
    \begin{subfigure}[b]{0.4\textwidth}
        \centering
        \includegraphics[width=\textwidth]{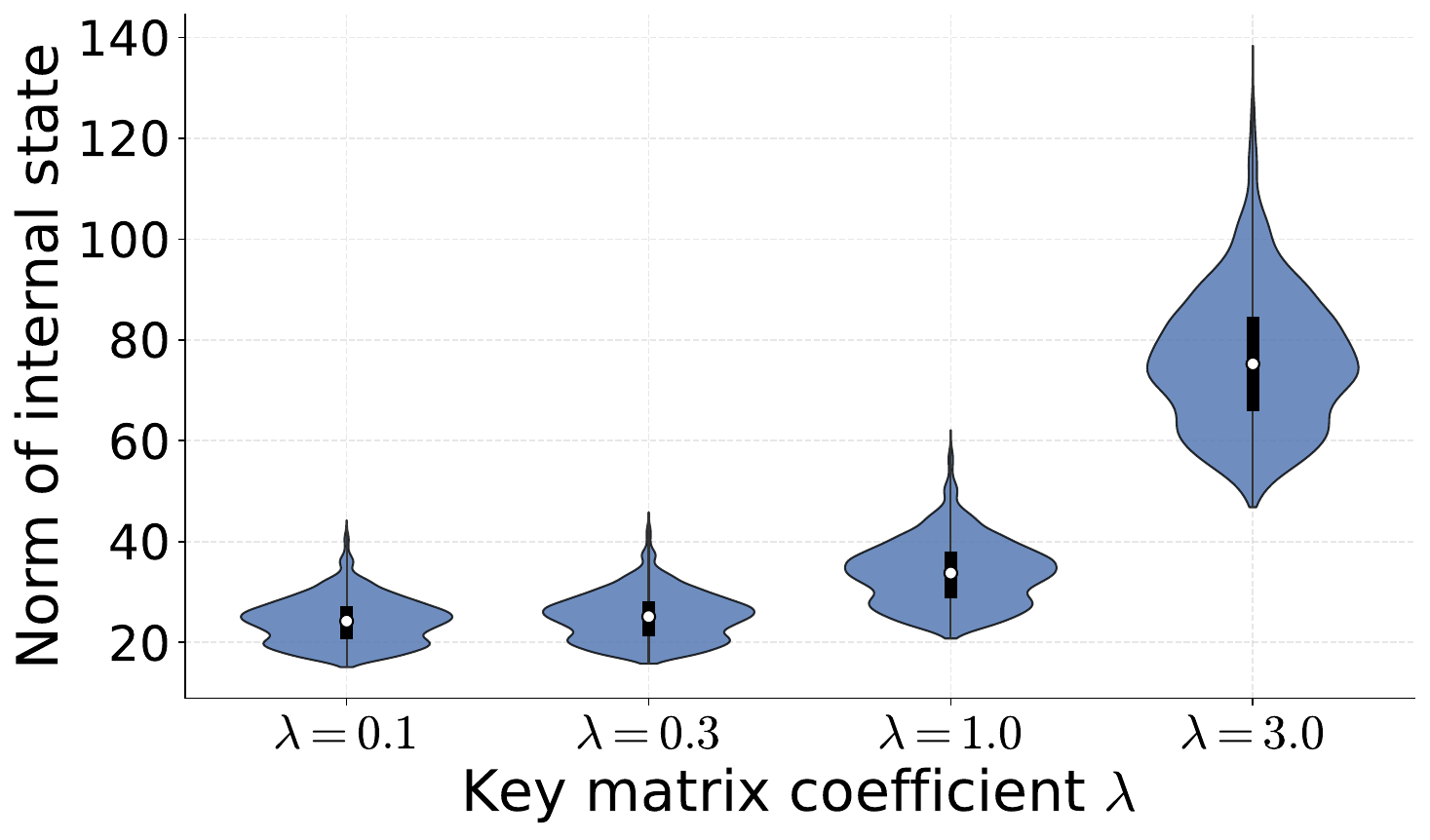}
        \caption{\small Norm distribution of internal states.}
        \label{fig:lambda_norm_dist}
    \end{subfigure}
    \caption{Impact of $\lambda$ on Qwen2.5-14B-Instruct.}
    \label{fig:lambda_impact}
\end{figure}

\begin{table*}[ht]
\centering
\scriptsize  
\caption{Accuracy and privacy of \texttt{AloePri} under various models and datasets. Privacy is evaluated with VMA on PUPA dataset.}
\label{tab:utility}
\scalebox{0.8}{
\begin{tabular}{cc|ccccc|ccccc|ccc}
\toprule \toprule

\multicolumn{2}{c}{\multirow{2}{*}{Model}} & \multicolumn{5}{c}{Plaintext accuracy $\uparrow$} & \multicolumn{5}{c}{\texttt{AloePri} accuracy $\uparrow$} & \multicolumn{3}{c}{\texttt{AloePri} privacy against VMA $\downarrow$} \\
\cmidrule(lr){3-7} \cmidrule(lr){8-12} \cmidrule(lr){13-15} 

& & MMLU & C-Eval & HumanEval & PIQA & IFEval  & MMLU & C-Eval & HumanEval & PIQA & IFEval & TTRSR(\%) & PIIRSR(\%) & BLEU-4 \\
\midrule
\multirow{2}{*}{R1-Distill}
& 14B 
& 78.91 & 86.18 & 96.34 & 88.14 & 72.83
& 79.2 & 85.93 & 93.9 & 88.79 & 72.27
& 12.36 & 2.37 & 0.52 \\
& 32B 
& 83.5 & 88.34 & 97.56 & 88.79 & 73.94
& 83.61 & 87.84 & 95.12 & 88.08 & 73.57
& 10.10 & 2.63 & 0.78 \\ \midrule
\multirow{2}{*}{Qwen3}
& 14B 
& 87.64 & 87.35 & 94.51 & 89.72 & 86.14 
& 84.57 & 87.12 & 95.12 & 90.26 & 85.40
& 25.05 & 1.62 & 1.72 \\ 
& 32B 
& 89.61 & 89.0 & 98.78 & 90.75 & 84.29
& 86.24 & 87.64 & 96.95 & 89.5 & 81.15
& 19.64 & 1.23 & 1.41 \\ \midrule
\multirow{1}{*}{Llama3} 
& 8B 
& 65.91  & 50.16 & 55.49 & 80.74 & 68.58 
& 63.5 & 44.55 & 57.32 & 78.94 & 67.84 
& 2.57 & 0 & 0.84 \\ \midrule

\multirow{1}{*}{Qwen3-MoE-Instruct}
& 30B-A3B 
& 84.75 & 86.97 & 85.37 & 89.55 & 80.41
& 84.49 & 86.79 & 83.54 & 82.21 &  78.56
& 5.91 & 2.43 & 0.74 \\ \midrule
\multirow{1}{*}{DS-V3.1-Terminus (no\_think)}
& 671B 
&  91.72 & 87.67  & 94.51 & 83.84 & 86.69
& 90.67 & 85.65 & 95.12 & 80.96 & 84.66 
& 4.80 & 1.12 & 0.40 \\ 
\bottomrule \bottomrule
\end{tabular}}
\end{table*}

\subsection{\wqruan{Impact of Privacy Parameters}}
We then evaluated the accuracy (measured by C-Eval) and privacy (against the VMA attack) of \texttt{AloePri} under different noise parameter settings. As depicted in Figure \ref{fig:privacy-utility},  $\alpha_{\text{e}}$ and $\alpha_{\text{h}}$ both play a critical role in preserving data privacy. For instance, the VMA attack can recover over 30\% tokens when $\alpha_{\text{e}}$ is set to $0.5$. The reason is that when the noise added to the weights of the embedding or model head is small, each pair of plain and obfuscated weights in Table \ref{tab:VMA_comb} differs almost only by row-wise and column-wise permutations. Consequently, VMA can easily recover the secret permutation \(\Pi\) through the approach of sorting followed by matching. Considering the trade-off between accuracy and privacy, we recommend setting $\alpha_{\text{e}}=1.0$ and $\alpha_{\text{h}}=0.2$ for practical applications.

In Figure \ref{fig:lambda_acc}, we evaluate the impact of $\lambda$ on model accuracy for both bfloat16 (BF16) and float32 (FP32) precisions. The results demonstrate that model accuracy degrades drastically at $\lambda = 3.0$ under the BF16 precision setting. Further observations of the norm distribution of internal states reported in Figure \ref{fig:lambda_norm_dist} reveal that this phenomenon arises because increasing $\lambda$ to enhance the obfuscation level of the key matrix simultaneously enlarges the distribution range of internal states, which makes numerical overflow more likely to occur during BF16-precision computations.

\subsection{Accuracy and Privacy under Various Models and Datasets}
To demonstrate the generality of \texttt{AloePri}, we tested the accuracy of \texttt{AloePri} on more models and datasets. We tested five commonly used datasets in the LLM benchmark (MMLU, C-Eval, HumanEval, PIQA, IFEval). For privacy, we used PUPA dataset, which is widely used to evaluate privacy levels of previous studies~\cite{siyan2024papillon}. 

As shown in Table~\ref{tab:utility}, \texttt{AloePri} can strongly protect input data while preserving model accuracy well for various models and tasks. Concretely, the accuracy loss of \texttt{AloePri} compared with plaintext model inference is less than 3\% in most datasets. Meanwhile, \texttt{AloePri} can effectively defend against VMA, i.e., PIIRSR of VMA on \texttt{AloePri} are less than 3\% for all models. These results further validate the generality of \texttt{AloePri} on models with various architectures.

% Meanwhile, we examine the privacy of \texttt{AloePri} across diverse models incorporating VMA, with the findings presented in Table \ref{tab:VMA_models}. The results indicate that \texttt{AloePri} maintains a stable privacy guarantee across models with varying parameter scales and architectural designs. Specifically, the VMA values of Qwen3-14B, Qwen3-32B, and the MoE-structured Qwen3-30B-A3B all fall within the range of xx–xx. This consistency arises from two key factors: \texttt{AloePri} exhibits compatibility with different model structures, and the noise magnitude it applies is correlated with the statistical characteristics of model weights. 

\begin{table}[ht]
\small  
\centering % 整个表格环境居中
% 第一个表格的minipage，宽度占文本宽度的48%
\begin{minipage}[t]{0.48\textwidth}
\centering
\caption{Ablative studies for internal state protection against ISA} 
\label{tab:ablation_isa} 
\scalebox{0.8}{
\begin{tabular}{c|cc}
\toprule 
\multirow{2}{*}{Applied mechanism}  & \multicolumn{2}{c}{TTRSR(\%) $\downarrow$} \\
\cmidrule(lr){2-3} 
& AttnScore & HiddenState \\ 
\midrule
Noise & 87.14  &  40.0 \\ \hline
Noise + KeyMat & 87.14  &  0.82 \\    \hline
Noise + KeyMat  + Head\&BlockPerm & 0.0 & 0.0  \\  
\bottomrule 
\end{tabular}}
\end{minipage}
\hfill % 两个minipage之间自动填充空白（居中对齐）
% 第二个表格的minipage，宽度占文本宽度的48%
\begin{minipage}[t]{0.48\textwidth}
\centering
\caption{Attack success rate of TFMA and SDA.}
\label{tab:freq_attack}
\scalebox{0.8}{
\begin{tabular}{cc|cc|c}
\toprule 
\multicolumn{2}{c|}{Attack setting} & \multicolumn{2}{c|}{TFMA(\%) $\downarrow$} & SDA $\downarrow$ \\
Pub. data & Priv. data & Top-10& Top-100 & BLEU-4 \\ \midrule
CCI3 & Huatuo26M & 0.14 & 1.01 & 0.01 \\
MedDialog & Huatuo26M & 0.43 & 3.83 &  0.55\\
Huatuo26M & Huatuo26M & 3.19 & 16.51 & 2.10 \\
\bottomrule 
\end{tabular}}
\end{minipage}
\end{table}

% \begin{table}[ht]
% \small  
% \centering
% \caption{Ablative studies for internal state protection against ISA} 
% \label{tab:ablation_isa} 
% \scalebox{0.8}{
% \begin{tabular}{c|cc}
% \toprule 
% \multirow{2}{*}{Applied mechanism}  & \multicolumn{2}{c}{TTRSR(\%) $\downarrow$} \\
% \cmidrule(lr){2-3} 
% & AttnScore & HiddenState \\ 
% \midrule
% Noise & 87.14  &  40.0 \\ \hline
% Noise + KeyMat & 87.14  &  0.82 \\    \hline
% % $Noise$ + $KeyMat$ + $BlockPerm$  &  &    \\   \hline
% Noise + KeyMat  + Head\&BlockPerm & 0.0 & 0.0  \\  
% \bottomrule 
% \end{tabular}}
% \end{table}

\subsection{Ablative Analysis of Internal States Protection}

We conducted ablative experiments to test the protective effect of \texttt{AloePri} on internal states (attention scores $AttnScore$ and decoder-layer hidden states $HiddenState$) against ISA. The obfuscation techniques tested include embedding noise $Noise$, key matrices $KeyMat$, attention block, and head permutation $Head\&BlockPerm$. 
As shown in Table \ref{tab:ablation_isa}, by combining embedding noise, key matrices, and attention block and head permutation, \texttt{AloePri} provides strong protection for internal states. Concretely, merely adding embedding noise is insufficient to prevent attackers from recovering tokens via internal states. After applying obfuscation of key matrices, the attacker can hardly exploit the hidden states output by decoder layers to perform ISA. However, attention scores remain vulnerable to ISA, as the attention scores generated by attention queries and keys are not obfuscated by the key matrices. Consequently, in \texttt{AloePri}, permutations of attention heads and blocks are introduced to further protect attention scores.

% \begin{table}[ht]
% \centering
% \small  
% \caption{Attack success rate of TFMA and SDA.}
% \label{tab:freq_attack}
% \scalebox{0.8}{
% \begin{tabular}{cc|cc|c}
% \toprule 
% \multicolumn{2}{c|}{Attack setting} & \multicolumn{2}{c|}{TFMA(\%) $\downarrow$} & SDA $\downarrow$ \\
% Pub. data & Priv. data & Top-10& Top-100 & BLEU-4 \\ \midrule

% CCI3 & Huatuo26M & 0.14 & 1.01 & 0.01 \\
% MedDialog & Huatuo26M & 0.43 & 3.83 &  0.55\\
% Huatuo26M & Huatuo26M & 3.19 & 16.51 & 2.10 \\
% \bottomrule 
% \end{tabular}}
% \end{table}

\subsection{Information Leakage of Token Frequency}
\label{sec:exp_token_freq}
We evaluated attackers’ ability to reconstruct private texts from token frequency via TFMA and SDA, defining three attack settings based on their prior knowledge of the client’s private data: \textbf{Zero-knowledge} (no client data information), where attackers leverage the public corpus CCI3 to recover the client’s medical domain dataset Huatuo26M-Lite; \textbf{Domain-aware} (access to client data domain information), where attackers use the medical domain dataset MedDiag to recover Huatuo26M-Lite; \textbf{Distribution-aware} (knowledge of the client data’s specific distribution), where attackers utilize Huatuo26M-Lite to recover a user dataset constructed from a different subset of the same dataset.

As shown in Table~\ref{tab:freq_attack}, in all three settings, \texttt{AloePri} can effectively prevent the attacker obtain private data through token frequency. Concretely, even when the attacker has prior knowledge of the user data distribution, it remains challenging to recover the client dataset by exploiting token frequency information. For instance, when using TFMA to directly perform matching based on token frequencies, the attacker’s Top-100 token recovery success rate is still no higher than 20\%. When employing the transformer-based SDA, the BLEU-4 score for the recovered text quality is only around 2, which is insufficient to form meaningful text.

% \begin{table}[ht]
% \centering
% \small  
% \caption{Runtime of Offline Model Obfuscation.}
% \label{tab:runtime}
% \scalebox{0.8}{
% \begin{tabular}{>{\centering\arraybackslash}p{4cm}|>{\centering\arraybackslash}p{2cm}}
% \toprule 
% Model & Time cost (min) \\
% \midrule
% R1-Distill-14B & 3.43 \\
% R1-Distill-32B & 9.28 \\
% Qwen3-MoE-30B-A3B & 4.58 \\
% Deepseek-V3.1-Terminus & 482.38 \\
% \bottomrule 
% \end{tabular}}
% \end{table}

% \begin{table}[ht]
% \centering
% \small  
% \caption{Online Inference TTFT (ms) and TPOT (ms)}
% \label{tab:efficiency}
% \scalebox{0.8}{
% \begin{tabular}{cc|cccc}
% \toprule 
% \multirow{2}{*}{Model} & \multirow{2}{*}{Concurrency} & \multicolumn{2}{c}{Plaintext} & \multicolumn{2}{c}{\texttt{AloePri}} \\
% & & TTFT & TPOT & TTFT & TPOT \\
% \midrule
% \multirow{2}{*}{R1-Distill-14B} & 1 & 19.49 & 6.93 & 20.56  & 6.72  \\
% & 4 & 33.34 & 7.41 & 36.64 & 7.22 \\ \midrule
% \multirow{2}{*}{Deepseek-V3.1-Terminus} & 1 & 166.08 & 19.14 & 172.04 & 19.63 \\
% & 4 & 185.11 & 21.42 & 199.97 & 22.41 \\
% \bottomrule 
% \end{tabular}}
% \end{table}

\begin{table}[ht]
\small  
\centering % 整个表格环境居中
% 第一个表格（耗时）的minipage，宽度适配列数较少的特点
\begin{minipage}[t]{0.38\textwidth}
\centering
\caption{Runtime of Offline Model Obfuscation.}
\label{tab:runtime}
\scalebox{0.8}{
\begin{tabular}{>{\centering\arraybackslash}p{4cm}|>{\centering\arraybackslash}p{2cm}}
\toprule 
Model & Time cost (min) \\
\midrule
R1-Distill-14B & 3.43 \\
R1-Distill-32B & 9.28 \\
Qwen3-MoE-30B-A3B & 4.58 \\
Deepseek-V3.1-Terminus & 482.38 \\
\bottomrule 
\end{tabular}}
\end{minipage}
\hfill % 自动填充两个表格间的空白
% 第二个表格（效率）的minipage，宽度适配列数较多的特点
\begin{minipage}[t]{0.58\textwidth}
\centering
\caption{Online Inference TTFT (ms) and TPOT (ms)}
\label{tab:efficiency}
\scalebox{0.8}{
\begin{tabular}{cc|cccc}
\toprule 
\multirow{2}{*}{Model} & \multirow{2}{*}{Concurrency} & \multicolumn{2}{c}{Plaintext} & \multicolumn{2}{c}{\texttt{AloePri}} \\
& & TTFT & TPOT & TTFT & TPOT \\
\midrule
\multirow{2}{*}{R1-Distill-14B} & 1 & 19.49 & 6.93 & 20.56  & 6.72  \\
& 4 & 33.34 & 7.41 & 36.64 & 7.22 \\ \midrule
\multirow{2}{*}{Deepseek-V3.1-Terminus} & 1 & 166.08 & 19.14 & 172.04 & 19.63 \\
& 4 & 185.11 & 21.42 & 199.97 & 22.41 \\
\bottomrule 
\end{tabular}}
\end{minipage}
\end{table}

\subsection{Efficiency of \texttt{AloePri}}

As shown in Table \ref{tab:runtime}, to evaluate the efficiency of \texttt{AloePri}, we test the runtime of offline model obfuscation and online inference on four models. Experimental results show that \texttt{AloePri} delivers practical efficiency for real-world deployment. For offline model obfuscation, it takes only about 10 minutes for a 32B dense model and about 8 hours for the 671B Deepseek-V3.1-Terminus. As this is an offline process, the runtime overhead is acceptable for practical use. For online inference, we compare \texttt{AloePri} with plaintext inference in text generation efficiency using \textit{Time to First Token (TTFT)} and \textit{Time Per Output Token (TPOT)} metrics. We test R1-Distill-14B with tensor parallelism set to 4, average 17-token prompts, 100-token generated texts and varying client request concurrency levels. Results in Table \ref{tab:efficiency} show that \texttt{AloePri} achieves identical online efficiency to plaintext inference.

\begin{figure}[ht]
    \centering
    \begin{subfigure}[b]{0.4\textwidth}
        \centering
        \includegraphics[width=\textwidth]{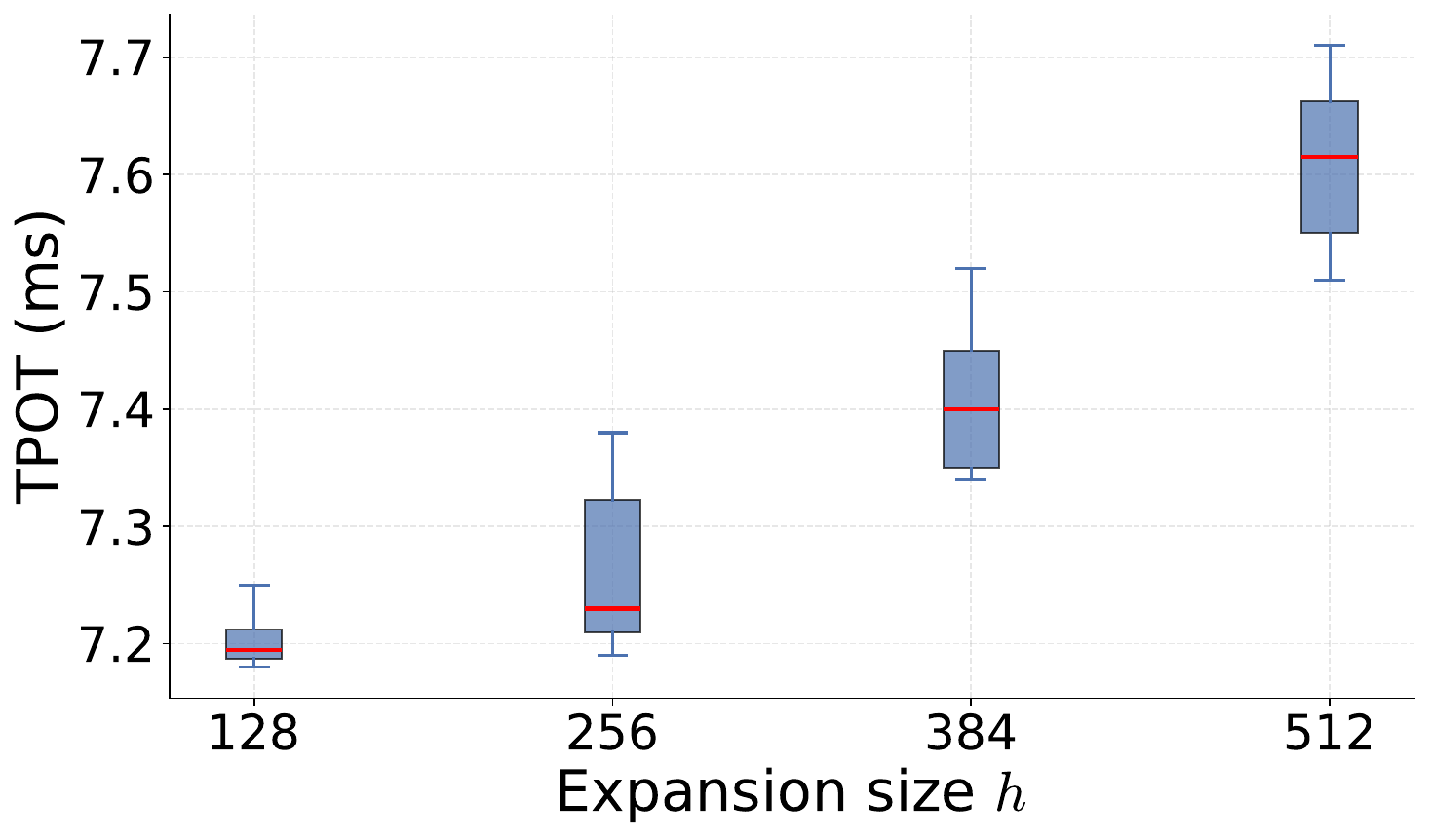}
        \caption{\small TPOT} 
        \label{fig:tpot_vs_h}
    \end{subfigure}
    % 第二个子图
    \begin{subfigure}[b]{0.4\textwidth}
        \centering
        \includegraphics[width=\textwidth]{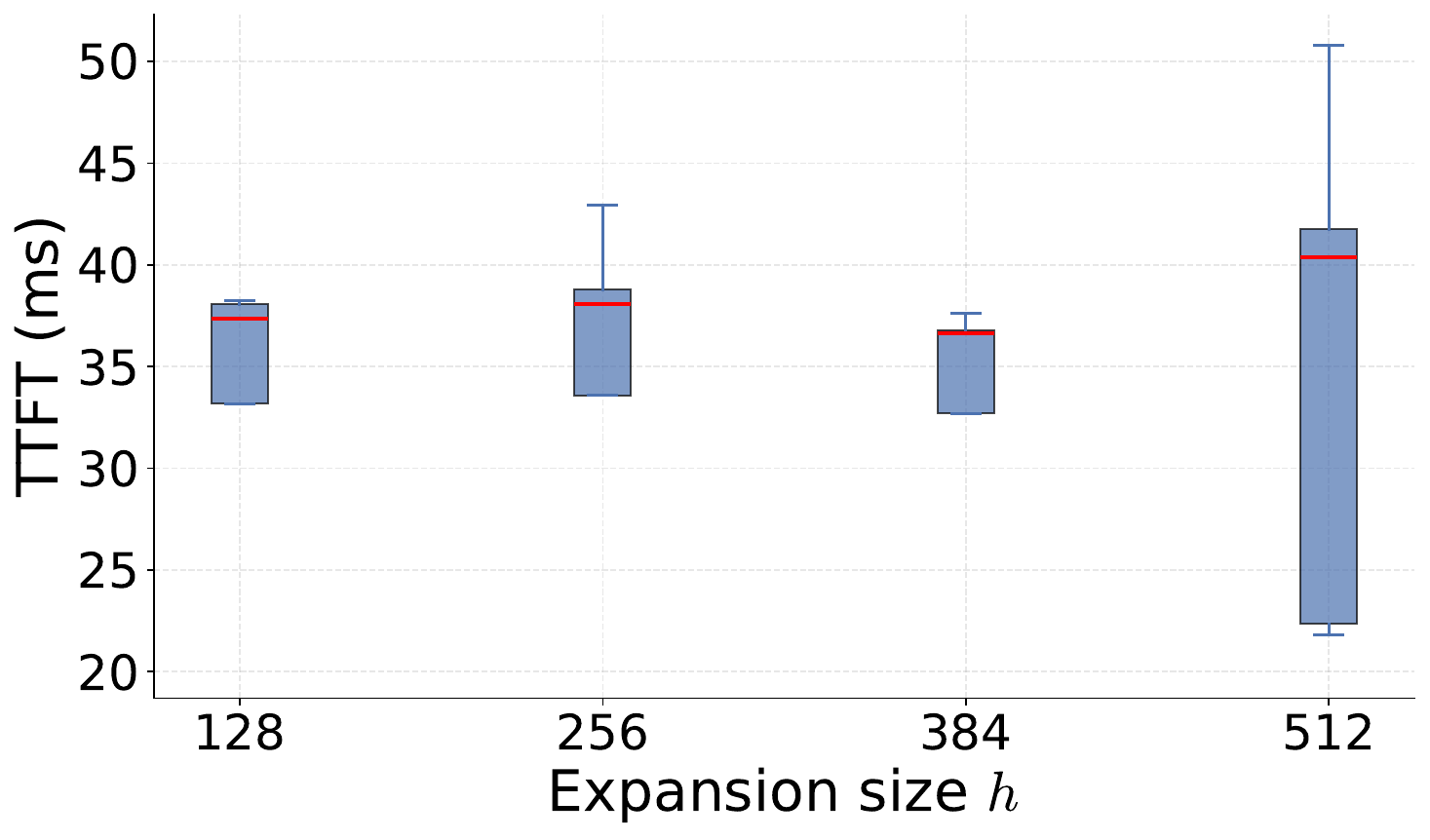}
        \caption{\small TTFT }
        \label{fig:ttft_vs_h}
    \end{subfigure}
    \caption{TPOT and TTFT vs. $h$ on R1-Distill-14B}
    \label{fig:latency_under_h}
\end{figure}

In Figure \ref{fig:latency_under_h}, we further investigate the impact of the expansion size on inference efficiency using R1-Distill-14B with request concurrency set to 4. Experimental results demonstrate that increasing the expansion size only slightly adds token latency, as modern LLM frameworks incorporate a variety of optimizations for parallel matrix computation. For example, the increase in TPOT is less than 10\% even when $h$ is scaled up to 512.

\section{Discussion and Future Work}
\label{sec:discussion}
\noindent\textbf{Generalization to Other Models.} We focus primarily on text-generative LLMs in this work, but our method’s design principles generalize to emerging Transformer variants and other neural network architectures (e.g., multi-modal LLMs, CNNs) under the covariant obfuscation mechanism. Most model architectures can be formulated as directed acyclic graphs (DAGs), with nodes (model components) connected sequentially or in parallel. Thus, the composition theorems of covariant obfuscation remain applicable for designing obfuscation schemes for these models, making carefully designed covariant obfuscation feasible for them.

\noindent\textbf{Protection of Model Weights.}  \texttt{AloePri} can be further extended to scenarios where only the server can access the model. In this case, CPU-TEE or FHE can be leveraged to allow the client to perform obfuscation without accessing the plaintext model. Taking FHE as an example: the server encrypts model weights via FHE and transmits them to the client, who obfuscates the encrypted weights and returns the results to the server. In this way, model obfuscation is achieved while keeping the model confidential from the client.

% \noindent\textbf{Privacy Analysis.}

% \noindent\textbf{Privacy Enhancement.} Comprehensive empirical studies validate \texttt{AloePri} against existing threats, though undiscovered attacks may still pose potential risks. Fortunately, the covariant obfuscation mechanism enables the integration of obfuscation techniques to further strengthen its privacy guarantees. For example, lightweight training can be incorporated into the model obfuscation pipeline to better resist inversion and obfuscation recovery attacks, as such attacks exploit similarities between obfuscated and plaintext model weights.

% From theoretical perspectives, we will further extend the composition theorems of covariant obfuscation to not only quantify the accuracy loss of composed covariant obfuscation mechanisms, but also their information leakage. Meanwhile, from application perspectives, we will develop privacy-preserving inference methods for more types of LLMs (e.g., multi-modal LLMs) to strengthen the privacy protection in additional LLM-driven application scenarios (e.g., medical image diagnosis).

\section{Conclusion }
In this paper, we propose covariant obfuscation and design a privacy-preserving LLM inference method \texttt{AloePri} for industrial applications. We theoretically demonstrate the superiority of covariant obfuscation over data obfuscation methods and propose three composition theorems for covariant obfuscation construction. Subsequently, by constructing covariant obfuscation for each model component and composing them based on composition theorems, \texttt{AloePri} can strongly protect user privacy in LLM inference while preserving model accuracy well.  Through comprehensive evaluations, we show that \texttt{AloePri} can maintain high model accuracy while effectively defending against representative attacks. For example, for Deepseek-V3.1-Terminus, with an accuracy loss smaller than 3.5\% and similar efficiency to plaintext inference, attackers can only recover less than 5\% of text tokens through inversion attacks. Our results demonstrate that utilizing internal structures of the model inference process based on covariant obfuscation could be a better way to design accurate, low-cost, and efficient privacy-preserving LLM inference methods.

%Bibliography
\bibliographystyle{unsrt}  
\bibliography{references}  

\appendix

\section{Notation Table}
\label{appx:notation}
The notations used in the paper is concluded in Table \ref{tab:notations}.

\begin{table}[ht]
\scriptsize
\centering
\caption{Notations and their Descriptions}
\label{tab:notations}
\renewcommand{\arraystretch}{1.2} % 调整行高，提升可读性
\scalebox{0.8}{
\begin{tabular}{p{2cm}|p{7cm}}
\hline
Notation & Description \\
\hline
$n$ & Number of vocabulary's tokens  \\
$d, d_{\text{head}}, d_{\text{ffn}}$ & Hidden size, attention head dimension, and FFN intermediate dimension \\
$h$ & Expansion size for the key matrix \\
$d'$ & Expanded hidden size \\
$m, m_{\text{kv}}$ & Number of attention heads, and number of attention key-value heads \\
$m_{\text{exp}}$ & Number of MoE experts \\
$\theta, \mathcal{V}$ & Weights and vocabulary of the language model \\
$L$ & Number of decoder layers \\
$W_{e}, W_{h}$ & Weights of the embedding layer and model head layer \\
$\boldsymbol{w}_{\text{norm}}$ & Weights of the layer normalization \\
$\omega_{\text{attn}}$ & Weights of an attention layer (including $W_{\text{query}}, W_{\text{key}}, W_{\text{value}}, W_{\text{out}}$) \\
$\omega_{\text{ffn}}$ & Weights of the FFN layer (including $W_{\text{gate}}, W_{\text{up}}, W_{\text{down}}$) \\
$\lambda$ & Coefficient for matrix sampling. \\
$\alpha_{\text{e}}, \alpha_{\text{h}}$ & Noise parameters for the embedding layer and model head layer \\
$\beta, \gamma$ & Maximum window size and sampling parameter for the attention block \\
$\tau, \Pi$ & Secret permutation of the vocabulary and its corresponding permutation matrix \\
$\hat{P}, \hat{Q}$ & Randomly generated key matrix and inverse key matrix \\
% $id_X$ & identity transformation on space $X$ \\
\hline
\end{tabular}}
\end{table}

\section{Proof of Covariant Obfuscation}
% \subsection{Proof of Theorem \ref{the:cov_obfs}}
% \label{appx:proof_of_cov_obfs}

% \begin{proof}
% Without loss of generality, we can let $X = \mathbb{Z}_n^{l}$ be
%  the input token space,  linear space $\Theta$ be the space of model weights,  $Y = \mathbb{Z}_n$ be the output token space. We denote $S_n$ as the token transformation group, which left-acts on $\Theta$. We define $H$ as the parameter transformation group, which right-acts on $\Theta$. 
%  The data-only obfuscation mechanism $D$ is defined by $D(x) := \sigma x$.
% We can construct a covariant obfuscation mechanism formulated as follows:
% \[
% \begin{array}{l}
% \phi_X(x) := \tau \sigma x, 
% \phi_\Theta(\theta) := \tau\theta, 
% \phi_Y(y) := \tau y, \\
% \psi_Y(y) := \tau^{-1}y, 
% \tilde{f} = f, 
% \end{array}
% \]
% where $ \tau \in S_n $ is a token permutation sampled from the uniform distribution. 
% Hence we have $\tilde{f}(\phi_X(x), \phi_\Theta(\theta))=f(\tau\sigma x, \tau\theta)=f(\sigma x, \theta)=f(D(x), \theta)$, hence $e_C=e_D$.
% On the other hand, $L_C=\mathcal{I}(x; \tau\sigma x, \tau\theta) < \mathcal{I}(x; \tau\sigma x, \tau \theta, \tau)=\mathcal{I}(x; \sigma x, \theta, \tau)= \mathcal{I}(x; \sigma x)=L_D$. 
% \end{proof}

% \subsection{Proof of Composition Theorems}
\label{appx:proof_of_comp_the}

We present the proofs of composition theorems of covariant obfuscation.

\noindent \textbf{Proof of Theorem \ref{the:sequantial}:}
\begin{proof}
We separately prove the commutation and de-obfuscation conditions of covariant obfuscation.

% {\footnotesize
% \[
% \begin{aligned}
% \tilde{h}(\phi_X(x), (\phi_\Theta(\theta), \phi_\Xi(\xi)))
% &= \tilde{g}(\tilde{f}(\phi_X(x), \phi_\Theta(\theta)), \xi) \\
% &\approx \tilde{g}(\phi_Y \circ f(x, \theta), \xi) \\
% &\approx \phi_Z \circ g(f(x, \theta), \xi) \\
% &= \phi_Z \circ h(x, (\theta, \xi)).
% \end{aligned}
% \]
% }

% \noindent Alternatively, the proof can be shown via diagram chasing as follows:
% \[
% \begin{CD}
% X @>f(\theta)>> Y @>g(\xi)>> Z \\
% @V\phi_X VV       @V\phi_Y VV       @V\phi_Z VV \\
% \tilde{X} @>\tilde{f}(\tilde{\theta})>> \tilde{Y} @>\tilde{g}(\tilde{\xi})>> \tilde{Z}
% \end{CD}
% \]

\noindent \textit{A. Commutation Condition:} If the error of \( C_1 \) is \( e_{C_1} \), the error of \( C_2 \) is \( e_{C_2} \), and \( \tilde{g}(\tilde{x}, \tilde{\xi}) \) satisfies the Lipschitz condition 
$
d(\tilde{g}(\tilde{x}_1, \tilde{\xi}), \tilde{g}(\tilde{x}_2, \tilde{\xi})) \leq M_g d(\tilde{x}_1, \tilde{x}_2)
$ with respect to \( \tilde{x} \), 
then we have
{
\[
\begin{aligned}
&d(\tilde{h}(\phi_X(x), (\phi_\Theta(\theta), \phi_\Xi(\xi))), \phi_Z \circ h(x, (\theta, \xi))) \\
&\leq d(\tilde{g}(\tilde{f}(\phi_X(x), \phi_\Theta(\theta)), \xi), \phi_Z \circ g(f(x, \theta), \xi)) \\
&\leq d(\tilde{g}(\tilde{f}(\phi_X(x), \phi_\Theta(\theta)), \xi), \tilde{g}(\phi_Y \circ f(x, \theta), \xi)) \\
&+ d(\tilde{g}(\phi_Y \circ f(x, \theta), \xi), \phi_Z \circ g(f(x, \theta), \xi))) \\
&\leq M_g d(\tilde{f}(\phi_X(x), \phi_\Theta(\theta)), \phi_Y \circ f(x, \theta)) \\
& + d(\tilde{g}(\phi_Y \circ f(x, \theta), \xi), \phi_Z \circ g(f(x, \theta), \xi))).
\end{aligned}
\]
}
Thus \(e_{C_2 \circ C_1} \leq M_g e_{C_1} + e_{C_2}\).

\noindent \textit{B. De-obfuscation Condition:} \( \psi_Z \circ \phi_Z = \text{id}_Z \) can be directly obtained from the de-obfuscation condition of \( (\phi_Y^2, \phi_\Xi, \phi_Z, \psi_Z, \tilde{g}) \).
\end{proof}

\noindent \textbf{Proof of Theorem \ref{the:parallel}:}

\begin{proof}
\noindent \textit{A. Commutation Condition}:
% \[
% \begin{CD}
% h: X \times (\Theta \times \Xi) @>>> Y \times Z \\
% @V(\phi_X, (\phi_\Theta, \phi_\Xi)) VV   @V(\phi_Y, \phi_Z) VV \\
% \tilde{h}: \tilde{X} \times (\tilde{\Theta} \times \tilde{\Xi}) @>>> \tilde{Y} \times \tilde{Z}
% \end{CD}
% \]
% {\footnotesize
% \[
% \begin{aligned}
% &\tilde{h}(\phi_X(x), (\phi_\Theta(\theta), \phi_\Xi(\xi))) = (\tilde{f}(\phi_X(x), \phi_\Theta(\theta)), \tilde{g}(\phi_X(x), \phi_\Xi(\xi))) \\
% &\approx (\phi_Y \circ f(x, \theta), \phi_Z \circ g(x, \xi)) = (\phi_Y, \phi_Z) \circ h(x, (\theta, \xi)).
% \end{aligned}
% \]
% }

 If the distance function \( d_{Y \times Z} \) on \( Y \times Z \) satisfies the control condition with respect to the distance functions on \( Y \) and \( Z \), then
{
\[
% \footnotesize
\begin{aligned}
&d_{Y \times Z}(\tilde{h}(\phi_X(x), (\phi_\Theta(\theta), \phi_\Xi(\xi))), (\phi_Y, \phi_Z) \circ h(x, (\theta, \xi))) \\
% =& d_{Y \times Z}((\tilde{f}(\phi_X(x), \phi_\Theta(\theta)), \tilde{g}(\phi_X(x), \phi_\Xi(\xi))), (\phi_Y \circ f(x, \theta), \phi_Z \circ g(x, \xi))) \\
\leq & d_Y(\tilde{f}(\phi_X(x), \phi_\Theta(\theta), \phi_Y \circ f(x, \theta)) \\
& + d_Z(\tilde{g}(\phi_X(x), \phi_\Xi(\xi)), \phi_Z \circ g(x, \xi))).
\end{aligned}
\]
}
% Thus, \( e_{C_1 \cup C_2} \leq e_{C_1} + e_{C_2} \).

\noindent \textit{B. De-obfuscation Condition:}
{
\[
% \footnotesize
(\psi_Y, \psi_Z) \circ (\phi_Y, \phi_Z) = (\psi_Y \circ \phi_Y, \psi_Z \circ \phi_Z) = (\text{id}_Y, \text{id}_Z) = \text{id}_{Y \times Z}.
\] 
}
\renewcommand{\qed}{} % 仅当前proof环境取消方框 
\end{proof}

\noindent \textbf{Proof of Theorem \ref{the:summation}:}

\begin{proof}
    \noindent \textit{A. Commutation Condition}:
% \[
% \begin{CD}
% h: X \times (\Theta \times \Xi) @>>> Y \\
% @V(\phi_X ,(\phi_\Theta, \phi_\Xi)) VV @V\phi_Y VV \\
% \tilde{h}: \tilde{X} \times (\tilde{\Theta} \times \tilde{\Xi}) @>>> \tilde{Y}
% \end{CD}
% \]
% {\footnotesize
% \[
% \begin{aligned}
%  \tilde{h}(\phi_X(x), &(\phi_\Theta(\theta), \phi_\Xi(\xi))) 
% =  \tilde{f}(\phi_X(x), \phi_\Theta(\theta)) + \tilde{g}(\phi_X(x), \phi_\Xi(\xi)) \\
% \approx & \phi_Y \circ f(x, \theta) + \phi_Y \circ g(x, \xi) = \phi_Y(f(x, \theta) + g(x, \xi)) \\
% = & \phi_Y \circ h(x, (\theta, \xi)).
% \end{aligned}
% \]
% }

If the distance function on \( Y \) satisfies translation invariance, then for any \( a_1, a_2, b_1, b_2 \in Z \), we have
{
\begin{align*}
d(a_1 + b_1, a_2 + b_2) 
\leq  & d(a_1 + b_1, a_2 + b_1) + d(a_2 + b_1, a_2 + b_2) \\
 = & d(a_1, a_2) + d(b_1, b_2).
\end{align*}
}
Thus,
{
\[
\begin{aligned}
&d_Y(\tilde{h}(\phi_X(x), (\phi_\Theta(\theta), \phi_\Xi(\xi))), \phi_Y \circ h(x, (\theta, \xi))) \\
=& d_Y( \tilde{f}(\phi_X(x), \phi_\Theta(\theta)) + \tilde{g}(\phi_X(x), \phi_\Xi(\xi)), \\ & \quad \phi_Y(f(x, \theta) + g(x, \xi)) ) \\
\leq & d_Y(\tilde{f}(\phi_X(x), \phi_\Theta(\theta)), \phi_Y \circ f(x, \theta))  \\
 & + d_Y(\tilde{g}(\phi_X(x), \phi_\Xi(\xi)), \phi_Y \circ g(x, \xi)).
\end{aligned}
\]
}
Therefore, \( e_{C_1 + C_2} \leq e_{C_1} + e_{C_2} \).

\noindent \textit{B. De-obfuscation Condition:} Directly follows from the de-obfuscation conditions of \( (\phi_X, \phi_\Theta, \phi_Z, \psi_Z, \tilde{f}) \) and \( (\phi_Y, \phi_\Xi, \phi_Z, \psi_Z, \tilde{g}) \).
\renewcommand{\qed}{} % 仅当前proof环境取消方框 
\end{proof}

\section{\qizhi{Proof of Theorem \ref{the:aloepri_rmdp}}}
% \begin{Qizhi}
\label{appx:proof_rmdp}

\subsection{R\'enyi Squared Metric Differential Privacy}
We first extend RmDP to R\'enyi squared metric Differential Privacy (RsmDP).
\begin{definition}[R\'enyi squared metric Differential Privacy]
\label{def:rm2dp}
Let $(\mathcal{X}, d)$ be a metric space. A randomized mechanism $\mathcal{M}\colon \mathcal{X} \to \mathcal{Y}$ satisfies \emph{$(\alpha,\epsilon,d^2)$-RsmDP} if, for any $x,x' \in \mathcal{X}$,
\[
D_\alpha\!\left(p_{\mathcal{M}(x)} \,\|\, p_{\mathcal{M}(x')}\right)
\le \epsilon \, d(x, x')^2,
\]
where $\alpha>1$, $\epsilon \ge 0$, and $p_{\mathcal{M}(x)}$ denotes the distribution of the mechanism output on input $x$.
\end{definition}

\subsection{Gauss Mechanism}
\noindent
\begin{theorem}[Gaussian Mechanism for RsmDP]
\label{Gaussian Mechanism}
Let $(\mathcal{X}, d)$ be a metric space, and let $f: \mathcal{X} \to \mathbb{R}^d$ be a query function satisfying
\[
L_2(f) := \max_{x, x' \in \mathcal{X}} \frac{\|\Sigma^{-1/2}(f(x) - f(x'))\|_2}{d(x,x')} < \infty,
\]
where $\Sigma$ is a symmetric positive definite matrix.

The Gaussian mechanism is defined by adding Gaussian noise $\mathcal{N}(0, \sigma^2 \Sigma)$ to the query output:
\[
\mathcal{M}(x) = f(x) + \mathcal{N}(0, \sigma^2 \Sigma).
\]
This mechanism satisfies $(\alpha, \epsilon, d^2)$-RsmDP with
\[
\epsilon = \frac{\alpha L_2(f)^2}{2\sigma^2}.
\]
\end{theorem}
\begin{proof}
Let $z := f(x) + \mathcal{N}(0, \sigma^2 \Sigma)$. For any $x, x' \in \mathcal{X}$, we have
\[
p_{z|x}(z) \sim \mathcal{N}(f(x), \sigma^2 \Sigma), \quad
p_{z|x'}(z) \sim \mathcal{N}(f(x'), \sigma^2 \Sigma).
\]
% According to the analysis of covariance obfuscation and information leakage under R\'enyi-metric differential privacy, 
Hence we obtain
\[
\begin{aligned}
D_\alpha(p_{z|x} \parallel p_{z|x'})
&= \frac{\alpha}{2} (f(x) - f(x'))^\top \sigma^{-2}\Sigma^{-1} (f(x) - f(x')) \\
&= \frac{\alpha}{2\sigma^2} \cdot \frac{\|\Sigma^{-1/2}(f(x)-f(x'))\|_2^2}{d(x,x')^2} \cdot d(x,x')^2 \\
&\leq \frac{\alpha L_2(f)^2}{2\sigma^2} \, d(x,x')^2
= \epsilon \, d(x,x')^2.
\end{aligned}
\]
\end{proof}

\noindent
\begin{theorem}[Gaussian Mechanism on Matrix Space for RsmDP]
\label{Gaussian Mechanism on Matrix Space}
Let $\Theta := M_{n,d}$. Define the randomized mechanism $M: O_n \to \Theta$ by
\[
M(g) := g(\theta + E),
\]
where each row of $E$ is independently and identically distributed as $\mathcal{N}(0, \sigma^2 I)$. Then the following statements hold:

\begin{enumerate}
    \item $M$ satisfies $(\alpha, \epsilon, d^2)$-RsmDP for any $\alpha > 1$, where
    \[
    L_2(f) := \max_{g,g' \in O_n,\; g \neq g'} \frac{\| g\theta - g'\theta \|_2}{d(g,g')}, 
    \qquad
    \epsilon = \frac{\alpha L_2(f)^2}{2\sigma^2}.
    \]

    \item If $d(g,g') := \|g - g'\|$, then $L_2(f) = \lambda_1$, where $\lambda_1$ is the largest singular value of $\theta$.

    \item If $d(g,g')$ is the geodesic distance \cite{etingof2024liegroupsliealgebras}, then
    \[
    L_2(f) \leq \sqrt{\frac{\lambda_1^2 + \lambda_2^2}{2}},
    \]
    where $\lambda_1 \ge \lambda_2$ are the largest and second-largest singular values of $\theta$.
\end{enumerate}
\end{theorem}

\begin{proof}
We prove each statement separately:
\begin{enumerate}
    \item This follows directly from Theorem \ref{Gaussian Mechanism} by setting $f(g) := g\theta$.

    \item Since
    \[
    \|g\theta - g'\theta\| \leq \|\theta\|_{OP}\,\|g - g'\|,
    \]
    where $\|\theta\|_{OP}$ denotes the operator norm of $\theta$, which is equal to its largest singular value $\lambda_1$, we have
    \[
    L_2(f) = \|\theta\|_{OP} = \lambda_1.
    \]

    \item Let $X := \log(g^{-1}g')$, which is a skew-symmetric matrix. Let $A := \theta$. Then
    \[
    \|g\theta - g'\theta\|^2 = \|g^{-1}g'A - A\|^2 = \|e^X A - A\|^2.
    \]
    We have
    \[
    \|e^X A - A\|^2 \leq \|XA\|^2 \leq \frac{\lambda_1^2 + \lambda_2^2}{2}\|X\|^2.
    \]
    Therefore,
    \[
    \|g\theta - g'\theta\|^2 \leq \frac{\lambda_1^2 + \lambda_2^2}{2}\|X\|^2
    = \frac{\lambda_1^2 + \lambda_2^2}{2}d(g,g')^2,
    \]
    which implies
    \[
    \frac{\|g\theta - g'\theta\|}{d(g,g')} \leq \sqrt{\frac{\lambda_1^2 + \lambda_2^2}{2}}.
    \]
    Since this holds for all $g,g'$, we conclude that
    \[
    L_2(f) \leq \sqrt{\frac{\lambda_1^2 + \lambda_2^2}{2}}.
    \]
\end{enumerate}
\end{proof}

\subsection{Exponential Mechanism}

\begin{theorem}[Exponential Mechanism for RmDP]
\label{Exponential Mechanism}
For a metric space $(X, d)$ where $X$ is a finite set, define the mechanism $\mathcal{M}: X \to X$ by $\mathcal{M}(x) = y$, where $y$ is a random variable on $X$ with probability distribution satisfying $p_{\mathcal{M}}(y|x) \propto e^{-\epsilon d(x,y)}$.
The mechanism $\mathcal{M}$ satisfies $(\alpha, \epsilon, d)$-RmDP for any $\alpha > 1$.
\end{theorem}

\begin{proof}
    We have
    \[
    \left( \frac{p_{\mathcal{M}}(y|x)}{p_{\mathcal{M}}(y|x')} \right)^{\alpha-1}
    = \exp\left[ \epsilon(\alpha-1)\left( d(x',y) - d(x,y) \right) \right]
    \leq \exp\left[ \epsilon(\alpha-1) d(x,x') \right].
    \]
    
    Therefore,
    \[
    D_\alpha\big(p_{\mathcal{M}}(y|x) \parallel p_{\mathcal{M}}(y|x')\big)
    = \frac{1}{\alpha} \log \int \left( \frac{p_{\mathcal{M}}(y|x)}{p_{\mathcal{M}}(y|x')} \right)^{\alpha-1} p_{\mathcal{M}}(y|x) dy
    \leq \frac{1}{\alpha-1} \log \exp\left[ \epsilon(\alpha-1) d(x,x') \right]
    = \epsilon d(x,x').
    \]
\end{proof}

\subsection{Permutation Metric on $S_n$}
Define the permutation metric $d_p(g,g')$ on $S_n$ as the minimum number $k$ of transpositions required to transform $g$ into $g'$, i.e.,
\[
g = g_0 \rightarrow g_1 \rightarrow \cdots \rightarrow g_k = g'.
\]
It is easy to see that the permutation metric on $\mathbb{Z}_n^l$ presented in Definition \ref{def:zn_perm_metric} is related to that on $S_n$ as follows:
\[
d(x, x') = \min_{g \in S_n:\; gx = x'} d_p(g, \mathrm{id}),
\]
where $\mathrm{id}$ denotes the identity permutation.

Moreover, we establish the relationship between the permutation metric and the geodesic metric in Theorem \ref{thm:geo_vs_perm}.

\begin{theorem}[Geodesic Metric vs. Permutation Metric]
\label{thm:geo_vs_perm}
We have
\begin{enumerate}
    \item For any $g, g' \in S_n$,
    \[
    d_g(g, g') \leq \pi \, d_p(g, g'),
    \]
    where $d_g$ denotes the geodesic metric on the unitary group $U_n$, and $d_p$ denotes the permutation metric on $S_n$.

    \item If a mechanism on $S_n$ satisfies $(\alpha, \epsilon, d_g^2)$-RsmDP, it also satisfies $(\alpha, \epsilon \pi^2, d_p^2)$-RsmDP.
\end{enumerate}
\end{theorem}

\begin{proof}
We prove each statement separately:
\begin{enumerate}
    \item Suppose $d_p(g,g') = k$. Then there exists a chain of transpositions
    \[
    g = g_0 \rightarrow g_1 \rightarrow \cdots \rightarrow g_k = g'.
    \]
    By the triangle inequality,
    \[
    d_g(g, g') \leq \sum_{i=0}^{k-1} d_g(g_i, g_{i+1})
    = \sum_{i=0}^{k-1} \pi
    = k\pi
    = \pi d_p(g, g').
    \]

    \item By assumption,
    \[
    D_\alpha\big(p(y|x) \,\|\, p(y|x')\big)
    \leq \epsilon \, d_g(x, x')^2.
    \]
    Combining this with the result in part (1), we obtain
    \[
    d_g(x,x')^2 \leq \pi^2 d_p(x,x')^2,
    \]
    which implies
    \[
    D_\alpha\big(p(y|x) \,\|\, p(y|x')\big)
    \leq \epsilon \pi^2 d_p(x, x')^2.
    \]
\end{enumerate}
\end{proof}

% \vspace{0.5em}
% \noindent
% \textbf{Theorem (Translation Invariance of Permutation Metric)}:
% For any $h \in G$, we have
% \[
% d(hg, hg') = d(gh, g'h) = d(g, g').
% \]

% \vspace{0.3em}
% \noindent
% \textit{Proof}:
% There exists a chain of transpositions $g = g_0 \rightarrow g_1 \rightarrow \cdots \rightarrow g_k = g'$
% if and only if there exists a chain $hg = hg_0 \rightarrow hg_1 \rightarrow \cdots \rightarrow hg_k = hg'$,
% and if and only if there exists a chain $gh = g_0 h \rightarrow g_1 h \rightarrow \cdots \rightarrow g_k h = g'h$.

\subsection{RsmDP for \texttt{AloePri}'s Offline Phase}

\begin{theorem}
The offline phase of \texttt{AloePri} satisfies $(\alpha,\epsilon,d^2)$-RsmDP for protecting the secret permutation $\tau$, where $\alpha=2$, $\epsilon=\epsilon_e+\epsilon_h$, and $d$ is the geodesic metric on $S_n \subset U_n$. Specifically,
\[
\epsilon_e = \frac{\alpha\big(\lambda_1^2(W_e) + \lambda_2^2(W_e)\big)}{4\sigma_e^2},
\]
where $\lambda_1(W_e) \ge \lambda_2(W_e)$ are the largest and second-largest singular values of $W_e$, and
\[
\epsilon_h = \frac{\alpha\big(\lambda_1^2(W_h) + \lambda_2^2(W_h)\big)}{4\sigma_h^2},
\]
where $\lambda_1(W_h) \ge \lambda_2(W_h)$ are the largest and second-largest singular values of $W_h$.
\end{theorem}

\vspace{0.5em}
\noindent
\textit{Proof}.
First, decompose the model parameter space into the embedding part, the middle part, and the model head part:
\[
\Theta = \Theta_e \oplus \Theta_{\text{middle}} \oplus \Theta_h.
\]
Correspondingly, let $\theta = (W_e, \theta_{\text{middle}}, W_h)$.
Then the offline mechanism can be expressed as:
\[
\mathcal{M}: S_n \to \Theta, \quad g \mapsto \left(g(\theta_e + E)P, \tilde{\mathcal{M}}(P), g(\theta_h + E)P^{-1}\right)
\]
for some random $P$ and sub-mechanism $\tilde{\mathcal{M}}$.

This mechanism can be decomposed into two parts: $\mathcal{M}(g) = \mathcal{M}_2 \circ \mathcal{M}_1(g)$, where
\[
\mathcal{M}_1(g) = \left(g(\theta_e + E), g(\theta_h + E)\right), \quad
\mathcal{M}_2(A,B) = \left(AP, \tilde{\mathcal{M}}(P), BP^{-1}\right).
\]

By the data processing inequality,
\[
D_\alpha\big(p_{\mathcal{M}}(z|g) \parallel p_{\mathcal{M}}(z|g')\big)
\leq D_\alpha\big(p_{\mathcal{M}_1}(y|g) \parallel p_{\mathcal{M}_1}(y|g')\big).
\]

On the other hand, by Theorem \ref{Gaussian Mechanism on Matrix Space}, we know that
\[
D_\alpha\big(p_{\mathcal{M}_1}(y|g) \parallel p_{\mathcal{M}_1}(y|g')\big) \leq \epsilon d^2(g,g'),
\]
where $\epsilon = \epsilon_e + \epsilon_h$, $d$ is the geodesic metric on $S_n \subset U_n$,
\[
\epsilon_e = \frac{\alpha\left(\lambda_1^2(W_e) + \lambda_2^2(W_e)\right)}{4\sigma_e^2},
\]
where $\lambda_1(W_e) \ge \lambda_2(W_e)$ are the largest and second-largest singular values of $W_e$;
\[
\epsilon_h = \frac{\alpha\left(\lambda_1^2(W_h) + \lambda_2^2(W_h)\right)}{4\sigma_h^2},
\]
where $\lambda_1(W_h) \ge \lambda_2(W_h)$ are the largest and second-largest singular values of $W_h$.

Therefore,
\[
D_\alpha\big(p_{\mathcal{M}}(z|g) \parallel p_{\mathcal{M}}(z|g')\big) \leq \epsilon d^2(g,g').
\]
\qed

\begin{corollary}
\label{offline_perm_metric}
The offline component satisfies $(\alpha, \epsilon, d^2)$-RDP, where $\alpha = 2$, $\epsilon = \pi^2 \cdot (\epsilon_e + \epsilon_h)$, and $d$ is the permutation metric on $S_n$.
\[
\epsilon_e = \frac{\alpha\left(\lambda_1^2(W_e) + \lambda_2^2(W_e)\right)}{4\sigma_e^2},
\]
where $\lambda_1(W_e) \ge \lambda_2(W_e)$ are the largest and second-largest singular values of $W_e$;
\[
\epsilon_h = \frac{\alpha\left(\lambda_1^2(W_h) + \lambda_2^2(W_h)\right)}{4\sigma_h^2},
\]
where $\lambda_1(W_h) \ge \lambda_2(W_h)$ are the largest and second-largest singular values of $W_h$.
\end{corollary}
\vspace{0.5em}
\noindent
\textit{Proof}.
Follows directly from the Theorem \ref{thm:geo_vs_perm}.

\qed

\subsection{Composability Theorem for Sub-Mechanisms on Token Sequence Space}

Let $\mathbb{Z}_n^l$ denote the space of token sequences of length $l$, and let $d(\cdot, \cdot)$ denote the permutation metric on $\mathbb{Z}_n^l$.

For the exponential mechanism $\mathcal{M}_1: \mathbb{Z}_n^l \to \mathbb{Z}_n^l$ and the mechanism $\mathcal{M}_2: S_n \to \Theta$, define the covariance obfuscated exponential mechanism $\mathcal{M}: \mathbb{Z}_n^l \to \mathbb{Z}_n^l \oplus \Theta$ as:
\[
\mathcal{M}(x) := \left(g\mathcal{M}_1(x), \mathcal{M}_2(g)\right),
\]
where $g \in S_n$ is uniformly distributed.

\medskip
\noindent
\textit{Note}: It is easy to see that $\mathcal{M}(x)$ can also be written as $\mathcal{M}(x) := \left(\mathcal{M}_1(gx), \mathcal{M}_2(g)\right)$.

\vspace{1em}
\noindent
\begin{theorem}[Composability for Sub-Mechanisms on Token Sequence Space]
\label{Composability-Sub-Mechanisms-Token}
If the exponential mechanism $\mathcal{M}_1: \mathbb{Z}_n^l \to \mathbb{Z}_n^l$ satisfies $(\alpha, \epsilon_1, d)$-privacy, and the mechanism $\mathcal{M}_2: S_n \to \Theta$ satisfies $\alpha$-$\epsilon_2$-$d^2$-privacy, where $\alpha = 2$, $\epsilon_1, \epsilon_2 \in \mathbb{R}^+$, and $d(\cdot, \cdot)$ is the permutation metric, then the mechanism $\mathcal{M}$ satisfies $\alpha$-$\epsilon$-$d$-privacy, where
\[
\epsilon =
\begin{cases}
\epsilon_1 - \dfrac{\epsilon_1^2}{4(n-1)\epsilon_2}, & \text{if } \epsilon_1 \leq 2(n-1)\epsilon_2, \\[6pt]
(n-1)\epsilon_2, & \text{otherwise}.
\end{cases}
\]
\end{theorem}

\medskip
\noindent
\textit{From this theorem we can see that}:
\begin{enumerate}
    \item If the sub-mechanisms consume less privacy budget (i.e., permutation recovery is inaccurate), then the privacy budget consumed by the parent mechanism will also decrease accordingly.
    \item Even if the sub-mechanisms consume a large privacy budget (i.e., permutation recovery is accurate), the privacy budget consumed by the parent mechanism is still strictly smaller than that of token obfuscation.
\end{enumerate}

\noindent
\textit{Proof}.
By the convexity of $\chi^2$-divergence, we have
\[
D_{\chi^2}\big(p_{\mathcal{M}}(z|x) \parallel p_{\mathcal{M}}(z|x')\big)
= D_{\chi^2}\big(E_{g,g'} p_{\mathcal{M}}(z|x,g) \parallel E_{g,g'} p_{\mathcal{M}}(z|x',g')\big)
\leq E_{g,g'} D_{\chi^2}\big(p_{\mathcal{M}}(z|x,g) \parallel p_{\mathcal{M}}(z|x',g')\big)
\]
for any joint distribution of random variables $(g,g')$ with uniform marginal distributions.

Since
\[
p_{\mathcal{M}}(z|x,g) = p_{\mathcal{M}_1}(z_1|g,x) \cdot p_{\mathcal{M}_2}(z_2|g),
\]
it follows that
\[
D_2\big(p_{\mathcal{M}}(z|x,g) \parallel p_{\mathcal{M}}(z|x',g')\big)
= D_2\big(p_{\mathcal{M}_1}(z_1|g,x) \parallel p_{\mathcal{M}_1}(z_1|g',x')\big)
+ D_2\big(p_{\mathcal{M}_2}(z_2|g) \parallel p_{\mathcal{M}_2}(z_2|g')\big).
\]
By assumption, $D_2\big(p_{\mathcal{M}_2}(z_2|g) \parallel p_{\mathcal{M}_2}(z_2|g')\big) \leq \epsilon_2 d^2(g,g')$, so
\[
D_2\big(p_{\mathcal{M}}(z|x,g) \parallel p_{\mathcal{M}}(z|x',g')\big)
\leq D_2\big(p_{\mathcal{M}_1}(z_1|g,x) \parallel p_{\mathcal{M}_1}(z_1|g',x')\big) + \epsilon_2 d^2(g,g').
\]

Furthermore, by the Theorem \ref{Exponential Mechanism} , we have
\[
D_2\big(p_{\mathcal{M}_1}(z_1|g,x) \parallel p_{\mathcal{M}_1}(z_1|g',x')\big)
\leq \epsilon_1 d(x, g' g x') = \epsilon_1 d(x, \Delta g x'),
\]
where $\Delta g = g^{-1} g'$.
Thus,
\[
D_2\big(p_{\mathcal{M}}(z|x,g) \parallel p_{\mathcal{M}}(z|x',g')\big)
\leq \epsilon_1 d(x, \Delta g x') + \epsilon_2 d^2(g,g') \quad \text{for any } \Delta g \in S_n.
\]

Assume $d(x,x') = r$. Then there exists $h \in S_n$ such that $x' = h x$ and $d(h,1) = r$.
Thus, there exists a chain of transpositions
\[
1 = h_0 \to h_1 \to \cdots \to h_r = h.
\]
For any $i = 0,1,\dots,r$, there exists $\Delta g_i := h_{r-i} h^{-1}$ such that
\[
d(x, \Delta g_i x') = d(x, h_{r-i} h^{-1} h x) = d(x, h_{r-i} x) = r - i,
\]
\[
d(g, g') = d(\Delta g_i, 1) = d(h_{r-i} h^{-1}, 1) = i.
\]

Therefore,
\[
D_2\big(p_{\mathcal{M}}(z|x,g) \parallel p_{\mathcal{M}}(z|x',g')\big)
\leq \epsilon_1 (r - i) + \epsilon_2 i^2 =: f_r(i).
\]
$f_r(i)$ is a quadratic function in $i$. Analysis shows that
\[
\min_{i=0}^r f_r(i) \leq \phi(r),
\]
where
\[
\phi(r) =
\begin{cases}
\epsilon_2 r^2, & \text{if } r \leq \left\lfloor \frac{\epsilon_1}{2\epsilon_2} \right\rfloor, \\[4pt]
f\left(\left\lfloor \frac{\epsilon_1}{2\epsilon_2} \right\rfloor\right) \approx \epsilon_1 r - \frac{\epsilon_1^2}{4\epsilon_2}, & \text{otherwise}.
\end{cases}
\]
Considering the range of $r$ is $[0, n-1]$, we have $\phi(r) \leq \epsilon r$, where
\[
\epsilon =
\begin{cases}
\epsilon_1 - \dfrac{\epsilon_1^2}{4(n-1)\epsilon_2}, & \text{if } \epsilon_1 \leq 2(n-1)\epsilon_2, \\[6pt]
(n-1)\epsilon_2, & \text{otherwise}.
\end{cases}
\]

Hence,
\[
D_2\big(p_{\mathcal{M}}(z|x,g) \parallel p_{\mathcal{M}}(z|x',g')\big) \leq \epsilon d(x,x').
\]

It follows that
\[
D_{\chi^2}\big(p_{\mathcal{M}}(z|x,g) \parallel p_{\mathcal{M}}(z|x',g')\big) \leq \exp\big(\epsilon d(x,x')\big) - 1,
\]
and thus
\[
D_{\chi^2}\big(p_{\mathcal{M}}(z|x) \parallel p_{\mathcal{M}}(z|x')\big)
\leq E_{g,g'} D_{\chi^2}\big(p_{\mathcal{M}}(z|x,g) \parallel p_{\mathcal{M}}(z|x',g')\big)
\leq \exp\big(\epsilon d(x,x')\big) - 1.
\]
Therefore,
\[
D_2\big(p_{\mathcal{M}}(z|x) \parallel p_{\mathcal{M}}(z|x')\big) \leq \epsilon d(x,x').
\]

\qed

\subsection{Proof of Theorem \ref{the:aloepri_rmdp}}
\textit{Proof.}
The offline phase satisfies $(\alpha,\epsilon_2,d^2)$-RsmDP with $\alpha=2$, where
\[
\epsilon_2 = \pi^2 \cdot (\epsilon_e + \epsilon_h),
\]
and $d$ denotes the permutation metric on $S_n$. The online phase is given by the covariance-obfuscated exponential mechanism. Therefore, by Theorem \ref{Composability-Sub-Mechanisms-Token}, the overall mechanism $\mathcal{M}$ satisfies $(\alpha,\epsilon,d)$-RmDP, where
\[
\epsilon =
\begin{cases}
\epsilon_1 - \dfrac{\epsilon_1^2}{4(n-1)\epsilon_2}, & \text{if } \epsilon_1 \le 2(n-1)\epsilon_2, \\
(n-1)\epsilon_2, & \text{otherwise}.
\end{cases}
\]
\qed

% \end{Qizhi}

\section{Experiment Details}
\subsection{Potential Threats}
\label{appx:attack}
\textbf{Vocabulary-Matching Attack (VMA)}.
\texttt{AloePri} protect privacy by mapping plaintext tokens of private text to obfuscated tokens via a permutation $\Pi$. With knowledge of the obfuscated model weights, attackers can attempt to map each obfuscated token back to its original plaintext token. This threat leverages the Vocabulary-Matching Attack (VMA) proposed by Thomas et al. \cite{thomashidden}, which is capable of recovering permutation matrices $Z_1$ and $Z_2$ from the known relationship $Y = Z_1 X Z_2$ (where $X$ is a known matrix and $Y$ is the observed obfuscated matrix).

\begin{table}[h]
    \centering
    \scalebox{0.8}{
    \begin{tabular}{p{4.5cm}|p{4.5cm}}
    \toprule \toprule
    Known $X$ & Observed $Y$ \\ \hline
    
    $W_{e}W_{h}$ & $\Pi W^{\star}_{\text{embed}}W_{h} \Pi^T$ \\ 

    $W_{e}W_{\text{query}} (W_{e}W_{\text{key}})^T$ & $ \Pi W^{\star}_{\text{embed}}W_{\text{query}} (\Pi W^{\star}_{\text{embed}}W_{\text{key}})^T$ \\

    $W_{e}W_{\text{gate}}$ & $\Pi W^{\star}_{\text{embed}}W_{\text{gate}} \hat{Z}_{\text{ffn}}$ \\ 

    $W_{e}W_{\text{up}}$ & $\Pi W^{\star}_{\text{embed}}W_{\text{up}}  \hat{S}_{\text{ffn}}\hat{Z}_{\text{ffn}}$ \\ 

    $W_{\text{down}}W_{h}$ & $\hat{Z}^{-1}_{\text{ffn}} \hat{S}^{-1}_{\text{ffn}}W_{\text{down}}W^{\star}_{\text{head}}   \Pi^T$ \\ 
    
    $W_{e}W_{\text{router}}$ & $\Pi W^{\star}_{\text{embed}} \text{norm}(W_{\text{router}}) \hat{Z}_{\text{router}}$ \\ 
    
    \bottomrule \bottomrule
    \end{tabular}}
    \caption{Weight combinations to recover $\Pi$ with VMA}
    \label{tab:VMA_comb}
\end{table}

The core idea of VMA is to first eliminate one permutation matrix through column-wise sorting, then recover the other one via neighbor matching. Specifically, let $\text{RowSort}(\cdot)$ denote the sorting operation along each row. it holds that $\text{RowSort}(Y) = \text{RowSort}(X Z_2)$. By comparing the rows of $\text{RowSort}(X Z_2)$ and $\text{RowSort}(X)$, attackers can recover $Z_1$. Table \ref{tab:VMA_comb} summarizes weight combinations in \texttt{AloePri} that exhibit a similar structural relationship to $Y = Z_1 X Z_2$. 
% Notably, even though an additional scaling matrix $S_{\text{ffn}}$ is used to obfuscate $W_{\text{up}}$ and $W_{\text{down}}$, column normalization can eliminate the scaling effect of each column, allowing these weights to still be exploited for VMA. 
For \texttt{AloePri}, weights from multiple decoder layers can be used to conduct VMA. Therefore, we adopt a voting mechanism to utilize VMA results of all layers.

\textbf{Invariant Attack (IA)}. Obfuscation mechanisms are vulnerable to this attack if they exhibit \textit{invariants} which remain unchanged during the obfuscation transformation \cite{lin2024inversion}. 
% By exploiting such invariants, attackers can derive the mapping relationship between plaintext and obfuscated data.
% without prior knowledge of the unknown transformations employed in the obfuscation mechanism.
Based on \texttt{AloePri}'s offline model obfuscation process, we present two types of IAs: Gate-IA and Attn-IA.

\textit{Gate-IA}: This attack leverages $W_{\text{gate}}$ to exploit statistical invariants. Let $\bm{e}$ and $\widetilde{\bm{e}}$ represent the plaintext and obfuscated embedding vectors, respectively. If no noise is added to the embedding layer ($\alpha_{\text{e}} = 0$), permutation preserves the statistical property of gate-weighted embeddings, resulting in $\text{Avg}(\bm{e} W_{\text{gate}}) = \text{Avg}(\widetilde{\bm{e}} \widetilde{W}_{\text{gate}})$. 
This average value serves as an invariant: it remains consistent between plaintext and obfuscated data. Attackers can use gate weights from different decoder layers to construct such invariant vectors, thereby reversing the token mapping relationship. To defend against Gate-IA, \texttt{AloePri} must add sufficient noise to embeddings by increasing $\alpha_{\text{e}}$ so that the statistical consistency of the invariant is disrupted.

\textit{Attn-IA}: This attack splits $W_{\text{query}}$ and $W_{\text{key}}$ into blocks (aligned with RoPE block structures) to exploit mathematical invariants. For analysis, we first ignore the head permutation and block permutation techniques introduced in Algorithms \ref{alg:gqa}. Let $W_{\text{query}}^{(i)}$ and $\widetilde{W}_{\text{query}}^{(i)}$ denote the plaintext and obfuscated weights of the $i$-th attention head, respectively. Define $\bm{q} = \bm{e} \cdot  W^{(i)}_{\text{query}} = [q_1, \cdots, q_{\frac{d_{head}}{2}}]$, $\widetilde{\bm{q}} = \widetilde{\bm{e}} \cdot \widetilde{W}_{\text{query}}^{(i)} = [\widetilde{q}_1, \cdots, \widetilde{q}_{d_{\text{head}}}]$, $Q = W_{e} W_{\text{query}}^{(i)}$, and $\widetilde{Q} = \widetilde{W}_{\text{embed}} \widetilde{W}_{\text{query}}^{(i)}$. For any block index $j$, if $e$ and $\widetilde{e}$ correspond to the same token, the mathematical relationship $e(Q^T[:, j] Q[:, j])^{-1} e^T = \widetilde{e} (\widetilde{Q}^T[:, j] \widetilde{Q}[:, j])^{-1} \widetilde{e}^T$ holds as an invariant. Attackers can iterate over all embeddings to verify this invariant relationship, enabling recovery of $\Pi$. 
To defend against Attn-IA, in addition to noise addition, $\texttt{AloePri}$ applies permutation over attention heads and blocks to further enhances security by breaking the block-wise consistency required for the invariant to hold.

\textbf{Internal State Attack (ISA)}.
% Since the entire forward computation process is evaluated on the server side, we test the success rate of the internal state inversion attack (ISA)  \cite{dong2025depthgivesfalsesense} on \texttt{AloePri}.
The attack optimizes input embeddings by leveraging the loss derived from hidden states. 
% This process enables attackers to recover plaintext embeddings, which in turn facilitates the recovery of private tokens. 
Specifically, the attacker first records the hidden states $State_1$ when clients request inference using their private input $X_1$. Subsequently, attackers randomly initialize $X_2$ and feeds it into the pretrained model to evaluate $State_2$. The attacker can use $State_1$ and $State_2$ to evaluate the loss, thereby optimizing $X_2$ to recover $X_1$.
In \texttt{AloePri}, model weights are perturbed with noise and transformation, so that all hidden states during the forward computation process are also noisy and transformed. Consequently, the recovered input data through ISA would differ significantly from the original model input, and cannot be directly used to recover private information.

\textbf{Inversion Model Attack (IMA)}.
With knowledge of the obfuscation mechanism, the attacker can train a model for embedding inversion \cite{kugler2021invbert}. During the training process, the attacker iterates over a public training dataset and generates obfuscated embeddings using the target obfuscation mechanism. With these obfuscated embeddings as inputs, the attacker trains the model to generate raw plaintext embeddings or token indices. In the experiments, we train a Qwen2 model with 2 decoder layers and 8 attention heads to invert obfuscated embeddings to plaintext token embeddings.
% \texttt{AloePri} exhibits resistance to the attack because the key matrices employed by the client remain hidden from the attacker. Consequently, an inversion model trained using a different set of key matrices fails to function on the client's obfuscated model.

\textbf{Token Frequency Leakage}. \texttt{AloePri} applies the deterministic token obfuscation based on the secret permutation. Therefore, an obfuscated text can be regarded as a substitution cipher and still retains the information about the statistical distribution of token frequencies in the plaintext. To investigate how token-level frequency information impacts privacy, we test the token frequency matching attack (TFMA) and the substitution deciphering attack (SDA) \cite{aldarrab-may-2021-sequence}. In both types of attacks, attackers attempt to recover substitution ciphertext based on word frequency information from a prior dataset. In TFMA, the attacker uses the prior dataset to count token frequencies, which are then used to match the frequency of ciphertext tokens in the client's data for token recovery. SDA introduces a recurrence encoding that converts substitution ciphers (1:1 and homophonic) into integer sequences by replacing every unique cipher symbol with an integer symbol standing for its frequency rank. A Transformer-based causal language model is trained on pairs of plain and cipher texts to learn symbol recurrence relations. Then the model is used to translate subsequent ciphertexts to plaintexts.

% Benefit from the characteristics of the LLM vocabulary space \texttt{AloePri} can defend against TFMA and SDA. In LLMs, the vocabulary typically contains hundreds of thousands of tokens. As a result, TFMA and SDA, which are based on word frequency information, are difficult to recover the client's private information from obfuscated tokens.

\subsection{Hyperparameter Setting}
\label{appx:hyperparam}
For all experiments, we fix the generation hyperparameters as follows: temperature = 0.65, top-k = 20, top-p = 0.95, maximum sequence length (max-seq-len) = 8192, and data type (dtype) = bfloat16. Table \ref{tab:exp_privacy_param} summarizes the privacy hyperparameters adopted in our experiments, including the noise coefficients for the embedding layer and model head ($\alpha_{e}, \alpha_{h}$), expansion size $h$, coefficient of the key matrix $\lambda$, and block permutation parameters $\beta, \gamma$ for attention layers.

\begin{table}[ht]
    \centering
    \scriptsize  
    \caption{Privacy hyperparameter settings of experiments.}
    \label{tab:exp_privacy_param}
    \begin{threeparttable}
    \begin{tabular}{>{\centering\arraybackslash}p{3cm}|cccccc}
        \toprule \toprule
        Experiment & $\alpha_{e}$ & $\alpha_{h}$ & $\lambda$ & $h$  &$\beta$ & $ \gamma$ \\ \midrule
        Fig. \ref{fig:privacy-utility} & $\triangle$\tnote{1} & $\triangle$ & 0.3 & 128 & 8 & $1e^3$\\
        Fig. \ref{fig:lambda_impact} & 1.0 & 0.2 & $\triangle$ & 128 & 8 & $1e^3$\\
        Fig. \ref{fig:latency_under_h} & 1.0 & 0.2 & 0.3 & $\triangle$ & 8 & $1e^3$ \\
        Tab. \ref{tab:baseline_comp}, \ref{tab:runtime}, \ref{tab:efficiency} & 1.0 & 0.2 & 0.3 & 128 & 8 & $1e^3$ \\
        Tab. \ref{tab:utility} & 0.5,1.0\tnote{2} & 0.2 & 0.3 & 128 & 8 & $1e^3$ \\
        Tab. \ref{tab:ablation_isa}  & 1.0 & 0.2 & $\triangle$ & 128 & $\triangle$ & $\triangle$ \\
        \bottomrule \bottomrule
    \end{tabular}

    \begin{tablenotes}
         \item[1] $\triangle$ denotes the hyperparameter tuned in the corresponding experiment.
        \item[2] $\alpha_e$ is set to 0.5 for Llama3-8B and 1.0 for all other models.
    \end{tablenotes}

    \end{threeparttable}
    
\end{table}

\subsection{Comparison Experiment}
\label{appx:imple_details}
\textbf{Hyperparameter Settings for Baselines}. In the experiment presented in Table \ref{tab:baseline_comp}, we tuned the hyperparameters of the baseline methods and \texttt{AloePri} to ensure a fair comparison. For \texttt{SANTEXT} and \texttt{RANTEXT}, we set their privacy budgets to $3$ and $30$, respectively. For \texttt{DP-Forward}, we adopted the LDP mode and set the privacy budget to $5$. Additionally, we set the sensitivity to $1$. We positioned the noise injection at the hidden states of the 1st decoder layer, thereby offloading the embedding layer and the first decoder layer (accounting for nearly 1B parameters) to the client. Meanwhile, the hidden states were not normalized before noise injection to avoid substantial accuracy loss in the Qwen model.

To reproduce \texttt{SGT}, we trained a 1.6B-parameter model following the Qwen2 architecture, which is used to generate the mean and covariance of embedding noise. This model comprises 2 decoder layers, 8 attention heads, a hidden size of 5120, and an intermediate size of 20480 for FFN. We trained the model on 1M text samples from the OpenOrca dataset and determined the training parameters through grid search. Specifically, we set the learning rate to $3e^{-4}$ and trained the model for 2 epochs. We also set the coefficients for accuracy loss, MI loss, absolute cosine loss, and normalization loss to $1.0$, $3e^{-5}$, $30$, and $1.0$, respectively.

\textbf{Benchmark Settings}. In the classification task SST2, we fine-tuned the model on the target task using LoRA \cite{hu2022lora}. We apply LoRA to attention weights and set $lora\_r = 64, lora\_alpha = 64$. We set the learning rate to $2e^{-5}$ for 3 training epochs with a sequence length of 128 and batch size of 64. We evaluated the classification task following the recommended application methods of each method. Specifically, for \texttt{SANTEXT} and \texttt{RANTEXT}, we first obfuscated the training and validation sets using these methods, followed by fine-tuning and inference. For \texttt{DP-Forward}, noise was injected to the outputs of decoder layers in both fine-tuning and inference processes. For  \texttt{SGT} and \texttt{AloePri}, we only applied obfuscation during the inference phase. 
% It is worth noting that applying obfuscation during the fine-tuning phase can enhance the inference accuracy of \texttt{SANTEXT}, \texttt{RANTEXT}, and \texttt{DP-Forward}.

Generation tasks do not involve fine-tuning or training processes. Thus, in the evaluation of generation tasks, we directly applied obfuscation methods to the inference process.

\end{document}